\documentclass[a4paper,onecolumn,11pt,accepted=2024-00-00]{quantumarticle}
\pdfoutput=1
\usepackage[utf8]{inputenc}
\usepackage[english]{babel}
\usepackage[T1]{fontenc}
\usepackage{lipsum}
\usepackage{bbm}
\usepackage{amsfonts}
\usepackage{amsmath}
\usepackage{bm}
\usepackage{graphicx}
\usepackage{latexsym}
\usepackage[margin=1.05in]{geometry}
  
\usepackage{color,graphicx}
\usepackage{array}
\usepackage{enumerate}
\usepackage{amsmath}
\usepackage{amssymb}
\usepackage{amsthm}
\usepackage{bbm} 
\usepackage{pgfplots}
\usepackage{pgf}
\usepackage{tikz}
\usetikzlibrary{patterns}
\usetikzlibrary{arrows.meta}
\usepgfplotslibrary{patchplots} 
\usetikzlibrary{pgfplots.patchplots} 
\pgfplotsset{width=9cm,compat=1.5.1}
\usepackage{diagbox}

\definecolor{myblue}{RGB}{0,0,128}
\definecolor{myblue2}{RGB}{0,32,96}
\definecolor{myblue3}{RGB}{0,64,64}
\definecolor{myblue4}{RGB}{0,96,32}
\definecolor{myblue5}{RGB}{0,128,0}
\definecolor{myblue6}{RGB}{32,96,0}
\definecolor{myblue7}{RGB}{64,64,0}
\definecolor{myblue8}{RGB}{96,32,0}
\definecolor{myblue9}{RGB}{128,0,0}

\usepackage{amsthm}

\newtheorem{theorem}{Theorem}[section]
   
\newtheorem{definition}[theorem]{Definition}
\newtheorem{lemma}[theorem]{Lemma}
\newtheorem{corollary}[theorem]{Corollary}
\newtheorem{remark}[theorem]{Remark}
\newtheorem{proposition}[theorem]{Proposition}
\newtheorem{example}[theorem]{Example} 

\usepackage{xcolor}
\usepackage{makeidx}
\usepackage[colorlinks=true,linkcolor=blue,anchorcolor=blue,citecolor=red,urlcolor=magenta]{hyperref}
\usepackage{caption}
\usepackage{subcaption}

\newcommand{\tr}{\operatorname{Tr}}

\newcommand{\bra}[1]{\langle #1 |}
\newcommand{\ket}[1]{| #1 \rangle}
\newcommand{\braket}[2]{\langle #1 | #2 \rangle}
\newcommand{\ketbra}[2]{| #1 \rangle\langle #2 |}

\newcommand\iprod[2]{\ensuremath{\langle#1|#2\rangle}}

\newcommand\ip[2]{\ensuremath{\langle#1|#2\rangle}}

\newcommand{\defeq}{\stackrel{\smash{\textnormal{\tiny def}}}{=}}

\newcommand{\kb}[1]{\ket{#1} \bra{#1}}

\def\C{\mathbb{C}}

\def\R{\mathbb{R}}

\def\Pos{\operatorname{Pos}}

\def\Herm{\operatorname{Herm}}
\def\Unitary{\operatorname{U}}

\begin{document}

\title{Tight bounds for antidistinguishability and circulant sets of pure quantum states}

\author{Nathaniel Johnston} 
\affiliation{Department of Mathematics \& Computer Science, Mount Allison University, Sackville, NB, Canada}  \email{njohnston@mta.ca} 
\homepage{https://njohnston.ca}

\author{Vincent Russo}
\affiliation{Unitary Fund}
\email{vincent@unitary.fund} 
\homepage{https://vprusso.github.io/}

\author{Jamie Sikora}
\affiliation{Department of Computer Science, Virginia Polytechnic Institute and State University, Blacksburg, Virginia, USA}
\email{sikora@vt.edu}
\homepage{https://sites.google.com/site/jamiesikora/}

\maketitle 

\begin{abstract}
    A set of pure quantum states is said to be antidistinguishable if upon sampling one at random, there exists a measurement to perfectly determine some state that was not sampled. We show that antidistinguishability of a set of $n$ pure states is equivalent to a property of its Gram matrix called $(n-1)$-incoherence, thus establishing a connection with quantum resource theories that lets us apply a wide variety of new tools to antidistinguishability. As a particular application of our result, we present an explicit formula (not involving any semidefinite programming) that determines whether or not a set with a circulant Gram matrix is antidistinguishable. We also show that if all inner products are smaller than $\sqrt{(n-2)/(2n-2)}$ then the set must be antidistinguishable, and we show that this bound is tight when $n \leq 4$.  
    We also give a simpler proof that if all the inner products are strictly larger than $(n-2)/(n-1)$, then the set cannot be antidistinguishable, and we show that this bound is tight for all $n$. 
\end{abstract}
  

\section{Introduction}\label{sec:intro}

A collection of pure quantum states $\{ \ket{\psi_0}, \ket{\psi_1}, \ldots, \ket{\psi_{n-1}} \}$ is called \emph{antidistinguishable}~\cite{caves2002conditions, leifer2014quantum,heinosaari2018antidistinguishability} if there exists a positive operator-valued measure $\{ M_0, M_1, \ldots, M_{n-1} \}$ such that 
\begin{equation} 
    \bra{\psi_i}M_i\ket{\psi_i} = 0, \text{ for all } i \in \{ 0, 1, \ldots, n-1 \}. 
\end{equation}
The outcome of the measurement can be interpreted as ruling out one of the $\ket{\psi_i}$ states. For example, if outcome $M_i$ occurs then we know for certain that $\ket{\psi_i}$ was not measured. 
The notion of antidistinguishability was introduced in~\cite{caves2002conditions} where it was referred to as \emph{post-Peierls incompatibility}. Antidistinguishability was later used as a key part in the proof of the PBR theorem~\cite{pusey2012reality}; a result that has significance to the foundations of quantum mechanics, and more specifically, significance to how one may interpret the reality of the quantum state. 

Antidistinguishability is also referred to as \emph{unambiguous quantum state exclusion}~\cite{bandyopadhyay2014conclusive}. This setting of quantum state exclusion (sometimes referred to as error-free quantum state elimination) has also found utility in the context of quantum communication~\cite{perry2015communication, heinosaari2019communication, havlivcek2020simple} as well as quantum cryptography where it has been used to reduce the need for long-term quantum memory for digital signature schemes~\cite{dunjko2014quantum} and to develop oblivious transfer protocols~\cite{amiri2021imperfect}.

In contrast to quantum state exclusion is the more well-established field of quantum state distinguishability that enjoys a rich history of study and has served to be foundational to the field of quantum information. In the setting of quantum state distinguishability, the goal is to determine what state one is given from a collection of quantum states. Whereas the setting of quantum state exclusion has the goal of determining which state one is \emph{not} given. Quantum state exclusion, and by proxy, the notion of antidistinguishability, has not been as thoroughly explored as quantum state distinguishability~\cite{bennett1999quantum, chefles2000quantum, walgate2000local, ghosh2001distinguishability, virmani2001optimal, walgate2002nonlocality, horodecki2003local, watrous2005bipartite, barnett2009quantum, bergou2010discrimination, qiu2010minimum, bae2015quantum}. 

One way in which to further our understanding of the notion of antidistinguishability is to determine under which conditions a collection of states is antidistinguishable. In~\cite{bandyopadhyay2014conclusive}, a necessary condition for antidistinguishability was provided as a function of the fidelity of the states in the collection. Similarly, in~\cite{heinosaari2018antidistinguishability}, the authors provided a sufficient condition for antidistinguishability based on algebraic properties of the states. 
In a recent work~\cite{mishra2023optimal}, optimal error exponents for antidistinguishability are given for the classical version of the problem and they also provide bounds for the quantum case, leaving an exact expression as an open problem. 
Antidistinguishability also appears in the study of quantum contextuality~\cite{PhysRevA.101.062113,PRXQuantum.3.020307}. 

In~\cite{havlivcek2020simple}, the authors conjectured that if a collection of $d$ states each of dimension $d$ satisfied an inequality based on $d$, then the states are antidistinguishable.  
The validity of this conjecture would imply the existence of an improved separation between a classical and quantum communication task~\cite{havlivcek2020simple} as well as a strengthening of the PBR theorem~\cite{montina2014lower, montina2014necessary}. This conjecture is known to be true for $d=2$ and for $d=3$~\cite{caves2002conditions} and also had some amount of numerical evidence to suggest that it might be true for higher dimensions as well~\cite{havlivcek2020simple}. However, a counterexample to the conjecture for $d=4$ was presented in~\cite{russo2023inner}. While this disproved the conjecture, the counterexample was not optimal and it was not clear whether the conjecture could be reframed or salvaged. We provide an optimal disproof of the conjecture for $d=4$ in Example~\ref{ex:d4_lb_tight} as well as a correction to the conjecture in Corollary~\ref{cor:antidist_by_IP}. In particular, our correction is a trivial-to-compute sufficient condition for antidistinguishability of a family of states based on their inner products.

In order to establish our results, we explore how antidistinguishability of a collection of pure quantum states can be determined by their Gram matrix. 
In some sense, considering the Gram matrix in this context is a natural thing to do and is inspired by the following references on the quantum change point problem~\cite{sentis2016quantum, sentis2017exact}. In particular, we establish a novel connection between antidistinguishability and quantum resource theories: we show in Theorem~\ref{thm:n_1_incoh} that a collection of pure states is antidistinguishable if and only if their Gram matrix is ``$(n-1)$-incoherent'' \cite{RBC18}. Since numerous properties of $(n-1)$-incoherent states are known \cite{LBT19,LSLL21,LM14,ZGY21}, this provides a wide array of new tools that can be used to investigate antidistinguishability, and we use a result from \cite{JMPP22} to establish our correction to the conjecture. We also establish numerous other necessary and sufficient conditions for antidistinguishability along the way that are of independent interest. 
Finally, we note that if the Gram matrix is circulant, then we derive an exact characterization of its antidistinguishability.

\subsection{Structure of the paper}\label{sec:structure}

We start in Section~\ref{sec:preliminaries} by presenting some mathematical background material that is required to present our results. 
In particular, we introduce the mathematical basics of antidistinguishability in Section~\ref{sec:antidist_basics}, Gram matrices in Section~\ref{sec:gram}, circulant matrices in Section~\ref{sec:circulants}, and the concept of $(n-1)$-incoherence in Section~\ref{sec:n1incoh}.

We then proceed in Section~\ref{sec:technical_stuff} to establish some of our more technical results.
In Section~\ref{sec:gram_sdp}, we develop a new (somewhat simpler than previously known) semidefinite program for checking antidistinguishability of a set of quantum states that uses the set's Gram matrix as input. We then proceed in Section~\ref{sec:antidist_from_coherence} to show that antidistinguishability of a set is equivalent to $(n-1)$-incoherence of its Gram matrix.

The remaining sections of the paper are devoted to establishing bounds that can be used to determine (non-)antidistinguishability of a set in ways that are simpler to evaluate than semidefinite programs. In Section~\ref{sec:upper_bounds} we re-derive a trivial-to-compute necessary condition for antidistinguishability via our framework. In Section~\ref{sec:lower_bounds} we develop several new trivial-to-compute sufficient conditions for antidistinguishability, including a condition that is both necessary and sufficient for sets of pure states that have a circulant Gram matrix. 
Finally, we explore the question of how tight the conditions from Sections~\ref{sec:upper_bounds} and~\ref{sec:lower_bounds} are in Section~\ref{sec:tightness}. 

\section{Mathematical preliminaries}\label{sec:preliminaries}

Throughout this paper, $n$ and $d$ are positive integers, and $\C^d$ is a finite-dimensional complex Euclidean space with standard basis $\{\ket{0},\ket{1},\ldots,\ket{d-1}\}$. We use the notation $\Pos(\C^d)$, $\Herm(\C^d)$, and $\Unitary(\C^d)$ to represent the sets of positive semidefinite (PSD) operators, Hermitian operators, and unitary operators acting on $\C^d$, respectively. If $A,B \in \Herm(\C^d)$ then the notation $A \preceq B$ means that $B-A \in \Pos(\C^d)$. We use $I \in \Pos(\C^d)$ and $O \in \Pos(\C^d)$ for the identity and zero operators acting on $\C^d$ (or $I_n$ and $O_n$ if we want to emphasize their size), respectively. We often represent linear operators as matrices in the usual way via the standard basis but we index their entries starting at $0$ (so, for example, we use $A_{0,0} = \bra{0}A\ket{0}$ to denote the $(0,0)$-entry of a matrix $A$, which is the entry at $A$'s top-left corner).

We provide a brief introduction to the mathematics of quantum information theory; the interested reader should pursue any of a number of standard books \cite{NC00,watrous2018theory} for a more thorough treatment of the subject. A \emph{pure quantum state} is a column vector $\ket{\psi} \in \C^d$ with Euclidean norm equal to $1$. A \emph{positive operator-valued measure (POVM)} is a set $\left\{ M_i : 0 \leq i \leq n-1 \right\} \subset \Pos(\C^d)$ satisfying 
\begin{equation*}
    \sum_{i=0}^{n-1} M_i = I,
\end{equation*}
and we refer to an individual $M_i$ as a \emph{measurement}.

\subsection{Antidistinguishability}\label{sec:antidist_basics}

For a POVM $\{M_0, \ldots, M_{n-1}\} \subset \Pos(\C^d)$ and set of pure states $\{\ket{\psi_0}, \ldots, \ket{\psi_{n-1}}\} \subset \C^d$, the probability of obtaining outcome $0 \leq i \leq n-1$, given the state $\ket{\psi_i}$, can be calculated by
\begin{equation*}
    p(i) = \bra{\psi_i} M_i \ket{\psi_i},
\end{equation*}
where $\bra{\psi_i}$ is the conjugate transpose of $\ket{\psi_i}$. The set of states is \emph{antidistinguishable} if there exists a POVM such that $\bra{\psi_i} M_i \ket{\psi_i} = 0$ for all $0 \leq i \leq n-1$.

Whether a set is antidistinguishable or not can be determined by a semidefinite program (SDP)~\cite{bandyopadhyay2014conclusive,russo2023inner}; for a general introduction to semidefinite programming in the context of quantum information theory, see \cite{watrous2018theory}, for example. 
We note here that both the primal and dual problems below share the same optimal objective function values thanks to strong duality and, moreover, both problems attain an optimal solution. 
In particular, a set is antidistinguishable if and only if the optimal value of the following primal-dual pair of SDPs is equal to $0$:\vspace*{-0.6cm}
\begin{center}
\begin{equation}\label{eq:antidist-sdp}
  \begin{minipage}{2.4in}
    \centerline{\underline{Primal problem}}\vspace{-7mm}
    \begin{align*}
      \text{minimize:}\quad & \sum_{i = 0}^{n-1} \bra{\psi_i}M_i\ket{\psi_i} \\
      \text{subject to:}\quad & \sum_{i = 0}^{n-1} M_i = I, \\
      & M_0,\ldots,M_{n-1} \in \Pos(\C^d)
    \end{align*}
  \end{minipage}
  \hspace*{10mm}
  \begin{minipage}{2.2in}
    \centerline{\underline{Dual problem}}\vspace{-7mm}
    \begin{align}
      \text{maximize:}\quad & \tr(Y) \nonumber \\
      \text{subject to:}\quad & Y \preceq \kb{\psi_0}, \nonumber \\
       & \quad \ \vdots \nonumber \\
       & Y \preceq \kb{\psi_{n-1}}, \nonumber \\
      & Y \in \Herm(\C^d). \nonumber
    \end{align}
  \end{minipage}
  \end{equation}
\end{center}
Slightly more generally, if we divide the optimal value of this SDP by $n$ then we get exactly the minimum probability of incorrectly performing state exclusion on the set (i.e., determining a state from the set that we were \emph{not} given), when the states from the set are provided as input with uniform probability. The set is antidistinguishable if and only if this optimal probability of being incorrect is $0$.

\begin{example}\label{exam:trine}
    Consider the collection $\{\ket{\psi_0},\ket{\psi_1},\ket{\psi_2}\} \subset \C^2$ of the ``trine'' states:
    \[
        \ket{\psi_0} = \ket{0}, \quad \ket{\psi_1} = -\frac{1}{2}\big(\ket{0} + \sqrt{3}\ket{1}\big), \quad \ket{\psi_2} = -\frac{1}{2}\big(\ket{0} - \sqrt{3}\ket{1}\big).
    \]
    This set is well-known to be antidistinguishable \cite{graeme2018thesis} but not distinguishable (since the states are not orthogonal). Indeed, a measurement $\{M_0,M_1,M_2\}$ that antidistinguishes this set comes from simply choosing each $M_i$ to be (up to scaling) a rank-$1$ projection onto the orthogonal complement $\ket{\psi_i^\perp}$ of $\ket{\psi_i}$. In particular,
    \begin{align*}
        M_0 & = \frac{2}{3}\ketbra{\psi_0^\perp}{\psi_0^\perp} = \frac{2}{3}(I - \kb{\psi_0}),\\ 
        M_1 & = \frac{2}{3}\ketbra{\psi_1^\perp}{\psi_1^\perp} = \frac{2}{3}(I - \kb{\psi_1}),\\ 
        M_2 & = \frac{2}{3}\ketbra{\psi_2^\perp}{\psi_2^\perp} = \frac{2}{3}(I - \kb{\psi_2}),\\ 
    \end{align*}
    as illustrated in Figure~\ref{fig:trine}.
    
    \begin{figure}[ht]
        \centering
        \begin{tikzpicture}[scale=2.5]%
    		\coordinate (O) at (0,0);
    		\coordinate (P) at (1,0);
    		\coordinate (Q) at (-0.5,-0.86602540378);
    		\coordinate (R) at (-0.5,0.86602540378);
    		\coordinate (S) at (0,1);
    		\coordinate (T) at (0.86602540378,-0.5);
    		\coordinate (U) at (-0.86602540378,-0.5);
      

    		\draw[color=black] (P) arc (0:360:1);
    		
      
    		\draw[thick,dashed,color=red,-stealth] (O) -- (S) node[anchor=south]{$\ket{\psi_0^\perp}$};
    		\draw[thick,dashed,color=red,-stealth] (O) -- (T) node[anchor=north west]{$\ket{\psi_1^\perp}$};
    		\draw[thick,dashed,color=red,-stealth] (O) -- (U) node[anchor=north east]{$\ket{\psi_2^\perp}$};
      
    		\draw[thick,color=blue,-stealth] (O) -- (P) node[anchor=west]{$\ket{\psi_0}$};
    		\draw[thick,color=blue,-stealth] (O) -- (Q) node[anchor=north east]{$\ket{\psi_1}$};
    		\draw[thick,color=blue,-stealth] (O) -- (R) node[anchor=south east]{$\ket{\psi_2}$};
    		
    	\end{tikzpicture}
        \caption{The trine states $\{ \ket{\psi_0}, \ket{\psi_1}, \ket{\psi_2} \}$ on the unit circle in $\mathbb{R}^2$, indicated in solid blue above, are antidistinguishable as witnessed by the POVM $M_0 = \tfrac{2}{3}\kb{\psi_0^\perp}$, $M_1 = \tfrac{2}{3}\kb{\psi_1^\perp}$, $M_2 = \tfrac{2}{3}\kb{\psi_2^\perp}$, where $\{ \ket{\psi_0^\perp}, \ket{\psi_1^\perp}, \ket{\psi_2^\perp} \}$ are indicated in dashed red.}\label{fig:trine}
    \end{figure}
    
    Indeed, it is straightforward to check that $M_0 + M_1 + M_2 = I$, so this measurement is feasible in the primal SDP~\eqref{eq:antidist-sdp}, with orthogonality resulting in an objective value of $0$. More generally, any collection of $n$ pure states in $\C^2$ that have $\frac{1}{n}\sum_{i=0}^{n-1} \kb{\psi_i} = \frac{1}{2}I$ (i.e., pure states that are ``evenly distributed'' on the surface of the Bloch sphere) is antidistinguishable, since we can choose measurement operators that are orthogonal to each of them.
\end{example}

\subsection{Gram matrices}\label{sec:gram}

The \emph{Gram matrix} of a set $S = \{ \ket{\psi_0}, \ket{\psi_1}, \ldots, \ket{\psi_{n-1}} \} \subset \C^d$ is the matrix $G \in \Pos(\C^n)$ whose $(i,j)$-entry is $G_{i,j} = \ip{\psi_i}{\psi_j}$. It is straightforward to see that if $U \in \Unitary(\C^d)$ then $US := \{ U\ket{\psi_0}, U\ket{\psi_1}, \ldots, U\ket{\psi_{n-1}} \}$ has the same Gram matrix as $S$ (the converse of this statement is also true, but somewhat less obvious: if two sets of pure states $S, S^\prime \subset \C^d$ have the same Gram matrix then there exists $U \in \Unitary(\C^d)$ such that $S^\prime = US$ \cite[Section~2.2.3]{JohALA}).

We can write the Gram matrix succinctly as $G = \sum_{i,j=0}^{n-1} \braket{\psi_i}{\psi_j} \ket{i}\bra{j} = W^*W$, where 
\begin{equation}\label{eq:W_gram}
    W := \sum_{k=0}^{n-1} \ketbra{\psi_k}{k} 
\end{equation}
is the $d \times n$ matrix with $\ket{\psi_k}$ as its $k$-th column. A few properties of this $W$ matrix are convenient for our analysis. Firstly, we have $W\ket{k} = \ket{\psi_k}$ for all $0 \leq k \leq n-1$. Secondly, if the set $S$ is linearly independent, then $W$ has full column rank, in which case there exists an $n \times d$ matrix $V$ such that $VW = I_n$. In particular, this implies $V\ket{\psi_k} = \ket{k}$.

\subsection{Circulant matrices}\label{sec:circulants}

An $n \times n$ matrix $G$ is called \emph{circulant} if there exist scalars $g_0$, $g_1$, $\ldots$, $g_{n-1} \in \C$ so that
\[
    G = \begin{bmatrix}
        g_0 & g_1 & g_2 & \cdots & g_{n-2} & g_{n-1} \\
        g_{n-1} & g_0 & g_1 & \cdots & g_{n-3} & g_{n-2} \\
        g_{n-2} & g_{n-1} & g_0 & \cdots & g_{n-4} & g_{n-3} \\
        \vdots & \vdots & \vdots & \ddots & \vdots & \vdots \\
        g_2 & g_3 & g_4 & \cdots & g_0 & g_1 \\
        g_1 & g_2 & g_3 & \cdots & g_{n-1} & g_0
    \end{bmatrix}.
\]
There are two special matrices that are of particular importance when working with circulant matrices. In particular, we define
\begin{align}\label{eq:fourier}
    P := \begin{bmatrix}
        0 & 1 & 0 & 0 &\cdots & 0 \\
        0 & 0 & 1 & 0 &\cdots & 0 \\
        0 & 0 & 0 & 1 &\cdots & 0 \\
        \vdots & \vdots & \vdots & \vdots & \ddots & \vdots \\
        0 & 0 & 0 & 0 &\cdots & 1 \\
        1 & 0 & 0 & 0 &\cdots & 0
    \end{bmatrix} \quad \text{and} \quad F := \frac{1}{\sqrt{n}}\begin{bmatrix}
        1 & 1 & 1 & \cdots & 1 \\
        1 & \omega & \omega^2 & \cdots & \omega^{n-1} \\
        1 & \omega^2 & \omega^4 & \cdots & \omega^{2(n-1)}\\
        \vdots & \vdots & \vdots & \ddots & \vdots\\
        1 & \omega^{n-1} & \omega^{2(n-1)} & \cdots & \omega^{(n-1)^2}
    \end{bmatrix}
\end{align}
($P$ is a cyclic permutation matrix and $F$ is the Fourier matrix). The following characterization of circulant matrices is well-known (see \cite{CircBook}, for example):

\begin{proposition}\label{prop:circulant}
    Let $G$ be an $n \times n$ matrix, and let $P$ and $F$ be as in Equation~\eqref{eq:fourier}. The following are equivalent:
    \begin{enumerate}
        \item[a)] $G$ is circulant.

        \item[b)] $G = PGP^*$.

        \item[c)] $G$ is diagonalized by the Fourier matrix: $G = FDF^*$ for some diagonal matrix $D$.
    \end{enumerate}
\end{proposition}

Condition~(c) of the above proposition is particularly useful for us, as it tells us that we can construct a circulant Gram matrix with any (necessarily non-negative, adding up to $n$) eigenvalues that we like: just place those eigenvalues along the diagonal of a diagonal matrix $D$ and then $G = FDF^*$ will be a circulant Gram matrix with those eigenvalues.

If such a $G$ is positive semidefinite (and thus Hermitian, so $g_j = \overline{g_{n-j}}$ for all $1 \leq j \leq n-1$) with $g_0 = 1$ then it is the Gram matrix of some set of pure states $S = \{\ket{\psi_0},\ket{\psi_1},\ldots,\ket{\psi_{n-1}}\}$. In this case, $G$ being circulant corresponds to the inner products of the members of $S$ being invariant under cyclic permutations of the indices: $\braket{\psi_i}{\psi_j} = \braket{\psi_{i+1\pmod{n}}}{\psi_{j+1\pmod{n}}}$ for all $i,j$. This motivates the following definition:

\begin{definition}\label{defn:circulant_states}
    We say that a set of pure quantum states $S = \{\ket{\psi_0},\ket{\psi_1},\ldots,\ket{\psi_{n-1}}\}$ is \emph{circulant} if it has either of the following equivalent properties:
    \begin{enumerate}
        \item[a)] The Gram matrix of $S$ is circulant.

        \item[b)] There exists a pure state $\ket{\psi}$ and a unitary matrix $U$ with the property that $U^n = I$ such that $S = \{\ket{\psi}, U\ket{\psi}, U^2\ket{\psi}, \ldots, U^{n-1}\ket{\psi}\}$.

        \item[c)] $\braket{\psi_i}{\psi_j} = \braket{\psi_{i+1\pmod{n}}}{\psi_{j+1\pmod{n}}}$ for all $0 \leq i,j \leq n-1$.
    \end{enumerate}
\end{definition}

We note that sets of quantum states with property~(b) above are sometimes called \emph{symmetric}~\cite{DA12} or geometrically uniform~\cite{howard2012qudit, dalla2015optimality}. The fact that that property is equivalent to property~(a) is proved in \cite[Proposition~3.12]{DM16}, where it was furthermore shown that $\ket{\psi}$ can be chosen to belong to $\R^n$ and have non-negative entries, and $U$ can be chosen to be $U = \mathrm{diag}(1,\omega,\omega^2,\ldots,\omega^{n-1})$, where $\omega = \exp(2\pi i/n)$ is a primitive $n$-th root of unity.

\subsection{\texorpdfstring{$\mathbf{(n-1)}$}{(n-1)}-incoherence}\label{sec:n1incoh}

One of our main results is the fact that antidistinguishability of a set of pure states is equivalent to a certain notion from the theory of quantum resources:

\begin{definition}[\cite{LM14,Spe15}]\label{defn:kincoh}
    Let $k$ be a positive integer. Then $X \in \Pos(\C^n)$ is called \emph{$k$-incoherent} if there exists a positive integer $m$, a set $S = \{\ket{\psi_0}, \ket{\psi_1}, \ldots, \ket{\psi_{m-1}}\} \subset \C^n$ with the property that each $\ket{\psi_i}$ has at most $k$ non-zero entries, and real scalars $c_0$, $c_1$, $\ldots$, $c_{m-1} \geq 0$ for which
    \begin{align}\label{eq:kincoh_decomp}
        X = \sum_{j=0}^{m-1} c_j\ketbra{\psi_j}{\psi_j}.
    \end{align}
\end{definition}

Strictly speaking, the term ``$k$-incoherent'' is typically only applied to positive semidefinite operators with trace $1$. However, the trace does not substantially affect any properties of $k$-incoherence, and it is more convenient for us to omit the trace restriction. In pure mathematics, a $k$-incoherent operator is sometimes said to have \emph{factor width at most $k$} \cite{boman2005factor}. Informally, $X$ is $k$-incoherent exactly when it can be written as a convex combination of positive semidefinite matrices, each of which is identically zero outside of a single $k \times k$ principal submatrix. For example, a positive semidefinite matrix is $1$-incoherent if and only if it is diagonal, and every $n \times n$ PSD matrix is $n$-incoherent.

We are particularly interested in the case when $k = n-1$, so we restrict our attention to $(n-1)$-incoherence for the rest of the paper. When $n = 2$ the $(n-1)$-incoherent operators are (as mentioned earlier) exactly those that are PSD and diagonal. When $n \geq 3$ this set of matrices is somewhat more complicated, but membership in it can be determined efficiently by semidefinite programming \cite{RBC18}. In particular, $X$ is $(n-1)$-incoherent if and only if there exist matrices $F_0,F_1,\ldots, F_{n-1} \in \Pos(\C^n)$ for which $G = \sum_{j=0}^{n-1} F_j$ and $\bra{i}F_i\ket{i} = 0$ for all $i$: each $F_i$ is equal to $\sum_j c_j \ketbra{\psi_j}{\psi_j}$ where the sum ranges over the pure states $\ket{\psi_j}$ that equal $0$ in their $i$-th entry, and conversely each $F_i$ with $\bra{i}F_i\ket{i} = 0$ has a spectral decomposition made up of pure states whose $i$-th entry equals $0$. 

For example, decompositions like the one below can be found quickly by computer software, thus certifying $(n-1)$-incoherence:
\begin{equation}
    \begin{bmatrix}
        2 & 1 & 2 \\
        1 & 2 & -1 \\
        2 & -1 & 5
    \end{bmatrix} = \begin{bmatrix}
        0 & 0 & 0 \\
        0 & 1 & -1 \\
        0 & -1 & 1
    \end{bmatrix} + \begin{bmatrix}
        1 & 0 & 2 \\
        0 & 0 & 0 \\
        2 & 0 & 4
    \end{bmatrix} + \begin{bmatrix}
        1 & 1 & 0 \\
        1 & 1 & 0 \\
        0 & 0 & 0
    \end{bmatrix},
\end{equation} 
where the $3$ matrices on the right are what we called $F_0$, $F_1$, and $F_2$ above. Computing a spectral decomposition of these $3$ matrices would then give a pure state $(n-1)$-incoherence decomposition of the form in equation~\eqref{eq:kincoh_decomp}.

The set of all $(n-1)$-incoherent $X \in \Pos(\C^n)$ is a closed convex cone inside the real vector space $\Herm(\C^n)$, so it admits separating hyperplanes. That is, for every $\widetilde{X} \in \Pos(\C^n)$ which is \emph{not} $(n-1)$-incoherent, there exists $Y \in \Herm(\C^n)$ (a separating hyperplane) with the property that $\tr(XY) \geq 0$ for all $(n-1)$-incoherent $X \in \Pos(\C^n)$ and $\tr(\widetilde{X}Y) < 0$. The following definition describes these separating hyperplanes more explicitly:

\begin{definition}[\cite{blekherman2022hyperbolic,JMPP22}]\label{defn:k_fac_pos}
    We say that $Y \in \Herm(\C^n)$ is \emph{$(n-1)$-locally PSD} if it has any of the following equivalent properties:
    
    \begin{enumerate}
        \item[a)] $\tr(XY) \geq 0$ for all $(n-1)$-incoherent $X \in \Pos(\C^n)$.
        
        \item[b)] $\bra{\psi}Y\ket{\psi} \geq 0$ for all pure states $\ket{\psi} \in \C^n$ with at most $n-1$ non-zero entries.
        
        \item[c)] Every $(n-1) \times (n-1)$ principal submatrix of $Y$ is positive semidefinite.
    \end{enumerate}
\end{definition}

In other words, the sets of $(n-1)$-incoherent operators and $(n-1)$-locally PSD operators are dual cones of each other (see \cite{BV04} for an introduction to dual cones). Given an operator $X \in \Pos(\C^n)$ that is not $(n-1)$-incoherent, it is straightforward to use semidefinite programming to find an $(n-1)$-locally PSD operator $Y$ for which $\tr(XY) < 0$, thus certifying non-$(n-1)$-incoherence of $X$. For example, if
\[
    X = \begin{bmatrix}
        1 & 1 & 1 \\
        1 & 1 & 1 \\
        1 & 1 & 1
    \end{bmatrix} \quad \text{and} \quad Y = \begin{bmatrix}
        1 & -1 & -1 \\
        -1 & 1 & -1 \\
        -1 & -1 & 1
    \end{bmatrix}
\]
then it is straightforward to show that every $(n-1) \times (n-1) = 2 \times 2$ principal submatrix of $Y$ is PSD, so $Y$ is $(n-1)$-locally PSD, but $\tr(XY) = -3 < 0$, so $X$ is not $(n-1)$-incoherent (despite being PSD).

We close this section by showing that circulant matrices play particularly well with $(n-1)$-incoherence and $(n-1)$-locally positive semidefiniteness. The following result shows that when investigating $(n-1)$-incoherence of circulant matrices, it suffices to consider $(n-1)$-locally PSD matrices that are also circulant:

\begin{lemma}\label{lem:circulant_n1_inc}
    Suppose $X \in \Herm(\C^n)$ is circulant. Then we have that $X$ is $(n-1)$-incoherent if and only if $\tr(XY) \geq 0$ for all $n \times n$ circulant $(n-1)$-locally PSD matrices $Y$.
\end{lemma}

\begin{proof}
    The ``only if'' direction follows immediately from Definition~\ref{defn:k_fac_pos}: if $X$ is $(n-1)$-incoherent then $\tr(XY) \geq 0$ for \emph{all} (not necessarily circulant) $(n-1)$-locally PSD matrices $Y$. We thus just need to prove the ``if'' direction.

    To this end, consider the linear map $P_{\textup{C}} : \Herm(\C^n) \rightarrow \Herm(\C^n)$ defined by
    \[
        P_{\textup{C}}(X) \defeq \frac{1}{n}\sum_{j=0}^{n-1} P^jX(P^j)^*,
    \]
    where $P$ is the permutation matrix from Equation~\eqref{eq:fourier}. It is straightforward to show that $P_{\textup{C}}(X)$ is circulant for all (not necessarily circulant) $X \in \Herm(\C^n)$. In fact, $P_{\textup{C}}$ is the orthogonal projection onto the $n$-dimensional subspace of $\Herm(\C^n)$ consisting of circulant matrices. Furthermore, if $X$ is $(n-1)$-locally PSD then so is each $P^jX(P^j)^*$, so $P_{\textup{C}}(X)$ is $(n-1)$-locally PSD too.

    Now suppose that $X$ is circulant (so $P_{\textup{C}}(X) = X$) and $\tr(XY) \geq 0$ for all circulant $(n-1)$-locally PSD matrices $Y$. Then for any (not necessarily circulant) $(n-1)$-locally PSD matrix $Z$ we have
    \[
        \tr(XZ) = \tr\big(P_{\textup{C}}(X)Z\big) = \tr\big(XP_{\textup{C}}(Z)\big) \geq 0,
    \]
    since $P_{\textup{C}}(Z)$ is circulant and $(n-1)$-locally PSD. It follows that $X$ is $(n-1)$-incoherent.
\end{proof}

\section{A reduced semidefinite programming formulation and technical results} 
\label{sec:technical_stuff}

We now present our technical results and mathematical framework for exploring antidistinguishability.

\subsection{An SDP formulation in terms of the Gram matrix}\label{sec:gram_sdp}

Our first result in this section is an alternate version of the semidefinite program~\eqref{eq:antidist-sdp} that is typically easier to work with (e.g., for finding explicit optimal solutions). This SDP uses the Gram matrix $G$ of the set of pure states $S$, rather than the states themselves:\vspace*{-0.6cm}
\begin{center} 
	\begin{equation} \label{eq:antidist-sdp2}
  \begin{minipage}{2.4in}
    \centerline{\underline{Primal problem}}\vspace{-7mm}
    \begin{align*}
      \text{minimize:}\quad & \sum_{i = 0}^{n-1} \bra{i} F_i \ket{i} \\
      \text{subject to:}\quad & \sum_{i = 0}^{n-1} F_i = G, \\
      & F_0,\ldots,F_{n-1} \in \Pos(\C^n) \nonumber
    \end{align*}
  \end{minipage}
  \hspace*{10mm}
  \begin{minipage}{2.2in} 
    \centerline{\underline{Dual problem}}\vspace{-7mm}
    \begin{align}
      \text{maximize:}\quad & \tr(XG) \nonumber \\
      \text{subject to:}\quad & X \preceq \kb{0}, \nonumber \\ 
       & \quad \ \vdots \nonumber \\ 
       & X \preceq \kb{n-1}, \nonumber \\ 
      & X \in \Herm(\C^n). \nonumber
    \end{align}
  \end{minipage}
  \end{equation}
\end{center}

Before proving that this semidefinite program has the same optimal value as the SDP~\eqref{eq:antidist-sdp}, we note that the primal and dual problems have a zero duality gap. This can be seen by the feasible primal solution $(F_0, F_1, \ldots, F_{n-1}) = (\frac{1}{n} G, \ldots, \frac{1}{n} G)$ and the strictly feasible dual solution $X = -I_n$. This also implies that the optimal value of this SDP is attained in the primal problem (so we do not need to consider sequences of primal feasible solutions converging onto our notion of antidistinguishability).

\begin{theorem}\label{thm:gram_sdp}
    The semidefinite programs~\eqref{eq:antidist-sdp} and~\eqref{eq:antidist-sdp2} have the same optimal value.
\end{theorem}

\begin{proof} 
    Let $G$ be the Gram matrix of the set $S \subset \C^d$, and define $W$ as in Equation~\eqref{eq:W_gram}, so that $G = W^*W$. 
    We prove this theorem by demonstrating a method of converting a feasible point of one SDP into a feasible point of the other SDP with the same objective function value.
    
    If $(M_0, M_1, \ldots, M_{n-1})$ is a feasible point of the SDP~\eqref{eq:antidist-sdp} then define $F_i = W^*M_iW$ for all indices $0 \leq i \leq n-1$. Then $(F_0, F_1, \ldots, F_{n-1})$ is a feasible point of the SDP~\eqref{eq:antidist-sdp2} since each $F_i$ is positive semidefinite and
    \[
        \sum_{i=0}^{n-1} F_i = \sum_{i=0}^{n-1} W^*M_iW = W^* \left(\sum_{i=0}^{n-1} M_i\right) W = W^*IW = W^*W = G.
    \]
    Furthermore, these feasible points give the same objective values in their respective SDPs since we have $W\ket{i} = \ket{\psi_i}$ and so
    \[
        \sum_{i=0}^{n-1} \bra{i}F_i\ket{i} = \sum_{i=0}^{n-1} \bra{i}W^*M_iW\ket{i} = \sum_{i=0}^{n-1} \bra{\psi_i}M_i\ket{\psi_i}.
    \]

    Conversely, if $(F_0, F_1, \ldots, F_{n-1})$ is a feasible point of the SDP~\eqref{eq:antidist-sdp2} then let $W^\dagger$ be the (Moore--Penrose) pseudoinverse of $W$ and define $M_i = (W^\dagger)^* F_i W^\dagger + \tfrac{1}{n}(I - WW^\dagger)$ for all $0 \leq i \leq n-1$. Then $(M_0, M_1, \ldots, M_{n-1})$ is a feasible point of the SDP~\eqref{eq:antidist-sdp} because:

    \begin{itemize}
        \item Each $M_i$ is positive semidefinite. To see this note that $WW^\dagger$ is the orthogonal projection onto $\mathrm{range}(W) = \mathrm{span}(S)$, so $I - WW^\dagger$ is positive semidefinite and thus $M_i$ is as well.

        \item If we recall the pseudoinverse property $(W^\dagger)^*W^*W = W$ then we see that
        \begin{align*}
            \sum_{i=0}^{n-1} M_i & = \sum_{i=0}^{n-1} \left((W^\dagger)^* F_i W^\dagger + \frac{1}{n}\big(I - WW^\dagger\big)\right) \\
            & = (W^\dagger)^*\left(\sum_{i=0}^{n-1} F_i\right) W^\dagger + \big(I - WW^\dagger\big) \\
            & = (W^\dagger)^*GW^\dagger + \big(I - WW^\dagger\big) \\
            & = (W^\dagger)^*W^*WW^\dagger + \big(I - WW^\dagger\big) \\
            & = WW^\dagger + \big(I - WW^\dagger\big) = I.
        \end{align*}
    \end{itemize}
    Furthermore, these feasible points give the same objective function values in their respective SDPs since
    \begin{align*}
        \sum_{i=0}^{n-1} \bra{\psi_i}M_i\ket{\psi_i} & = \sum_{i=0}^{n-1} \bra{\psi_i}\left((W^\dagger)^* F_i W^\dagger + \frac{1}{n}(I - WW^\dagger)\right)\ket{\psi_i} \\
        & = \sum_{i=0}^{n-1} \bra{\psi_i}(W^\dagger)^* F_i W^\dagger\ket{\psi_i} \\
        & = \sum_{i=0}^{n-1} \bra{i}(W^\dagger W)^* F_i (W^\dagger W)\ket{i} \\
        & = \sum_{i=0}^{n-1} \bra{i}F_i\ket{i},
    \end{align*}
    where the final equality follows from the fact that $W^\dagger W$ is the orthogonal projection onto $\mathrm{range}(W^*) \supseteq \mathrm{range}(F_i)$, so $(W^\dagger W)^*F_i(W^\dagger W) = F_i$. 
\end{proof} 

Thanks to the above theorem, the dimension $d$ that the set is embedded in is completely irrelevant when considering the antidistinguishability of pure states; all that matters is how many states there are and their inner products. 
We thus ignore the dimension $d$ in the remaining sections. 

\subsection{Antidistinguishability in terms of \texorpdfstring{$\mathbf{(n-1)}$}{(n-1)}-incoherence}\label{sec:antidist_from_coherence} 

We now show that if we only care about antidistinguishability itself, and not the optimal error probability when performing state exclusion, the SDP~\eqref{eq:antidist-sdp2} can be simplified even further, to the point that it coincides with $(n-1)$-incoherence of the states' Gram matrix:

\begin{theorem}\label{thm:n_1_incoh}
    Let $G$ be the Gram matrix of a set of $n$ pure states. Then the set is antidistinguishable if and only if $G$ is $(n-1)$-incoherent. 
\end{theorem}

\begin{proof}
    The primal version of the semidefinite program~\eqref{eq:antidist-sdp2} says that the set of states is antidistinguishable if and only if there exist $F_1, \ldots, F_n \in \Pos(\C^n)$ with $\sum_{i = 0}^{n-1} F_i = G$ and $\bra{i}F_i\ket{i} = 0$ for all $0 \leq i \leq n-1$. Since $\bra{i}F_i\ket{i} = 0$ is equivalent to the $i$-th row and column of $F_i$ being equal to $0$, this is equivalent to $(n-1)$-incoherence of $G$.    
\end{proof}

Thanks to Theorem~\ref{thm:n_1_incoh}, we can now show that a set of pure states is antidistinguishable by finding an $(n-1)$-incoherent decomposition of their Gram matrix $G$ (i.e., a way of writing $G = \sum_i F_i$, where each $F_i \in \Pos(\C^{n})$ has at least one row and column equal to $0$), and we can show that it is not antidistinguishable by finding an $(n-1)$-locally PSD matrix $Y$ for which $\tr(YG) < 0$. Both of these tasks can be carried out straightforwardly by semidefinite programming. While we could already determine antidistinguishability via semidefinite programming, this new SDP based on $(n-1)$-incoherence is a bit simpler and lets us derive several new explicit bounds on antidistinguishability in the upcoming sections.

We illustrate how to make use of Theorem~\ref{thm:n_1_incoh} with an example that determines exactly which equiangular bases of $\C^n$ are antidistinguishable.

\begin{example}\label{exam:equi_angle}
    Let $0 \leq \gamma \leq 1$ be a real number and let $S = \{\ket{\psi_0}, \ket{\psi_1}, \ldots, \ket{\psi_{n-1}}\}$ be such that $\ip{\psi_i}{\psi_j} = \gamma$ whenever $i \neq j$. In other words, if $\mathbf{1}$ is the all-ones vector then the Gram matrix of $S$ is
    \[
        G = I + \gamma (\mathbf{1}\mathbf{1}^T - I) = \begin{bmatrix}
            1 & \gamma & \gamma & \cdots & \gamma \\
            \gamma & 1 & \gamma & \cdots & \gamma \\
            \gamma & \gamma & 1 & \cdots & \gamma \\
            \vdots & \vdots & \vdots & \ddots & \vdots \\
            \gamma & \gamma & \gamma & \cdots & 1
        \end{bmatrix}.
    \]
    We claim that $S$ is antidistinguishable if and only if $\gamma \leq (n-2)/(n-1)$. To demonstrate this claim, we show that $G$ is $(n-1)$-incoherent if and only if $\gamma \leq (n-2)/(n-1)$ and then apply Theorem~\ref{thm:n_1_incoh}.

    To verify that $S$ is antidistinguishable when $\gamma \leq (n-2)/(n-1)$, define
    \[
        F_i := \left(\frac{1}{n-1}-\frac{\gamma}{n-2}\right)\big(I - \kb{i}\big) + \frac{\gamma}{n-2}\big(\mathbf{1}-\ket{i}\big)\big(\mathbf{1}-\ket{i}\big)^T \quad \text{for all} \quad 0 \leq i \leq n-1.
    \]
    It is clear that $F_i \in \Pos(\C^n)$ for all $i$ (since $\gamma \leq (n-2)/(n-1)$), and direct calculation shows that
    \begin{align}\begin{split}\label{eq:G_gamma_decomp}
        \sum_{i=0}^{n-1} F_i & = \sum_{i=0}^{n-1} \left(\left(\frac{1}{n-1}-\frac{\gamma}{n-2}\right)\big(I - \kb{i}\big) + \frac{\gamma}{n-2}\big(\mathbf{1}-\ket{i}\big)\big(\mathbf{1}-\ket{i}\big)^T\right) \\
        & = \left(1 - \frac{\gamma(n-1)}{n-2}\right)I + \gamma \left(\mathbf{1}\mathbf{1}^T + \frac{1}{n-2}I\right) \\
        & = G.
    \end{split}\end{align}
    It follows that $G$ is $(n-1)$-incoherent, since each $F_i$ has its $i$-th row and column equal to $0$ (i.e., Equation~\eqref{eq:G_gamma_decomp} is a decomposition of $G$ into a sum of $(n-1) \times (n-1)$ PSD blocks).
    
    Now suppose that $\gamma > (n-2)/(n-1)$. To verify that $S$ is not antidistinguishable, let
    \[
        X = (n-2)I - (\mathbf{1}\mathbf{1}^T - I).
    \]
    Since each $(n-1) \times (n-1)$ principal submatrix of $X$ is diagonally dominant, $X$ is $(n-1)$-locally PSD. However,
    \begin{align*}
        \tr(XG) & = \tr\Big( \big((n-2)I - (\mathbf{1}\mathbf{1}^T - I)\big)\big(I + \gamma (\mathbf{1}\mathbf{1}^T - I)\big) \Big) \\
        & = n(n-2) - n(n-1)\gamma,
    \end{align*}
    which is strictly negative (since $\gamma > (n-2)/(n-1)$). It follows that $G$ is not $(n-1)$-incoherent.  

    The above calculation shows that the SDP~\eqref{eq:antidist-sdp2} has its optimal value equal to $0$ if and only if we have $\gamma \leq (n-2)/(n-1)$. By simply running that SDP numerically (and dividing the result by $n$), we can furthermore find the optimal (i.e., minimal) error probability when performing state exclusion on this set of states, which is plotted in Figure~\ref{fig:santa_graph}.

    \begin{figure}[!htb]
    	\centering
    	\def\x{\noexpand\x}    
    	\begin{tikzpicture}[scale=1.0,yscale=5,xscale=12.5]
    		\foreach \x in {0.1, 0.2, 0.3, 0.4, 0.5, 0.6, 0.7, 0.8, 0.9, 1} \draw[color=gray!25] (\x,0) -- (\x,1.05);
    		\foreach \y in {0.25, 0.5, 0.75, 1} \draw[color=gray!25] (0,\y) -- (1.02,\y);
    		
    		\draw (0.1,0.5pt) -- (0.1,-1pt) node[anchor=north] {\footnotesize 0.1};
    		\draw (0.2,0.5pt) -- (0.2,-1pt) node[anchor=north] {\footnotesize 0.2};
    		\draw (0.3,0.5pt) -- (0.3,-1pt) node[anchor=north] {\footnotesize 0.3};
    		\draw (0.4,0.5pt) -- (0.4,-1pt) node[anchor=north] {\footnotesize 0.4};
    		\draw (0.5,0.5pt) -- (0.5,-1pt) node[anchor=north] {\footnotesize 0.5};
    		\draw (0.6,0.5pt) -- (0.6,-1pt) node[anchor=north] {\footnotesize 0.6};
    		\draw (0.7,0.5pt) -- (0.7,-1pt) node[anchor=north] {\footnotesize 0.7};
    		\draw (0.8,0.5pt) -- (0.8,-1pt) node[anchor=north] {\footnotesize 0.8};
    		\draw (0.9,0.5pt) -- (0.9,-1pt) node[anchor=north] {\footnotesize 0.9};
    		\draw (1,0.5pt) -- (1,-1pt) node[anchor=north] {\footnotesize 1.0};
      
    		\draw (0.2pt,0.25) -- (-0.4pt,0.25) node[anchor=east] {\footnotesize 0.25};
    		\draw (0.2pt,0.5) -- (-0.4pt,0.5) node[anchor=east] {\footnotesize 0.50};
    		\draw (0.2pt,0.75) -- (-0.4pt,0.75) node[anchor=east] {\footnotesize 0.75};
    		\draw (0.2pt,1) -- (-0.4pt,1) node[anchor=east] {\footnotesize 1.00};
    		
    
    		\draw[color=myblue] (0,0) -- (0.00401,0) -- (0.00803,0.00003) -- (0.01204,0.00007) -- (0.01606,0.00012) -- (0.02007,0.00020) -- (0.02409,0.00029) -- (0.02810,0.00039) -- (0.03212,0.00051) -- (0.03614,0.00065) -- (0.04015,0.00080) -- (0.04417,0.00097) -- (0.04818,0.00116) -- (0.05220,0.00136) -- (0.05621,0.00158) -- (0.06023,0.00181) -- (0.06425,0.00206) -- (0.06826,0.00233) -- (0.07228,0.00261) -- (0.07629,0.00291) -- (0.08031,0.00323) -- (0.08432,0.00356) -- (0.08834,0.00391) -- (0.09236,0.00427) -- (0.09637,0.00465) -- (0.10039,0.00505) -- (0.10440,0.00546) -- (0.10842,0.00589) -- (0.11243,0.00634) -- (0.11645,0.00680) -- (0.12046,0.00728) -- (0.12448,0.00777) -- (0.12850,0.00829) -- (0.13251,0.00881) -- (0.13653,0.00936) -- (0.14054,0.00992) -- (0.14456,0.01050) -- (0.14857,0.01109) -- (0.15259,0.01171) -- (0.15661,0.01233) -- (0.16062,0.01298) -- (0.16464,0.01364) -- (0.16865,0.01432) -- (0.17267,0.01502) -- (0.17668,0.01573) -- (0.18070,0.01646) -- (0.18472,0.01720) -- (0.18873,0.01797) -- (0.19275,0.01875) -- (0.19676,0.01954) -- (0.20078,0.02036) -- (0.20479,0.02119) -- (0.20881,0.02204) -- (0.21283,0.02291) -- (0.21684,0.02379) -- (0.22086,0.02469) -- (0.22487,0.02561) -- (0.22889,0.02654) -- (0.23290,0.02750) -- (0.23692,0.02847) -- (0.24093,0.02945) -- (0.24495,0.03046) -- (0.24897,0.03148) -- (0.25298,0.03253) -- (0.25700,0.03358) -- (0.26101,0.03466) -- (0.26503,0.03576) -- (0.26904,0.03687) -- (0.27306,0.03800) -- (0.27708,0.03915) -- (0.28109,0.04032) -- (0.28511,0.04150) -- (0.28912,0.04270) -- (0.29314,0.04393) -- (0.29715,0.04517) -- (0.30117,0.04643) -- (0.30519,0.04770) -- (0.30920,0.04900) -- (0.31322,0.05031) -- (0.31723,0.05165) -- (0.32125,0.05300) -- (0.32526,0.05437) -- (0.32928,0.05576) -- (0.33330,0.05717) -- (0.33731,0.05860) -- (0.34133,0.06005) -- (0.34534,0.06152) -- (0.34936,0.06301) -- (0.35337,0.06451) -- (0.35739,0.06604) -- (0.36140,0.06759) -- (0.36542,0.06915) -- (0.36944,0.07074) -- (0.37345,0.07235) -- (0.37747,0.07397) -- (0.38148,0.07562) -- (0.38550,0.07729) -- (0.38951,0.07898) -- (0.39353,0.08069) -- (0.39755,0.08241) -- (0.40156,0.08417) -- (0.40558,0.08594) -- (0.40959,0.08773) -- (0.41361,0.08954) -- (0.41762,0.09138) -- (0.42164,0.09323) -- (0.42566,0.09511) -- (0.42967,0.09701) -- (0.43369,0.09893) -- (0.43770,0.10088) -- (0.44172,0.10284) -- (0.44573,0.10483) -- (0.44975,0.10684) -- (0.45376,0.10888) -- (0.45778,0.11093) -- (0.46180,0.11301) -- (0.46581,0.11511) -- (0.46983,0.11724) -- (0.47384,0.11939) -- (0.47786,0.12156) -- (0.48187,0.12376) -- (0.48589,0.12598) -- (0.48991,0.12822) -- (0.49392,0.13049) -- (0.49794,0.13278) -- (0.50195,0.13510) -- (0.50597,0.13745) -- (0.50998,0.13981) -- (0.51400,0.14221) -- (0.51802,0.14463) -- (0.52203,0.14707) -- (0.52605,0.14954) -- (0.53006,0.15204) -- (0.53408,0.15456) -- (0.53809,0.15711) -- (0.54211,0.15969) -- (0.54613,0.16229) -- (0.55014,0.16493) -- (0.55416,0.16759) -- (0.55817,0.17027) -- (0.56219,0.17299) -- (0.56620,0.17573) -- (0.57022,0.17851) -- (0.57423,0.18131) -- (0.57825,0.18414) -- (0.58227,0.18700) -- (0.58628,0.18989) -- (0.59030,0.19281) -- (0.59431,0.19576) -- (0.59833,0.19875) -- (0.60234,0.20176) -- (0.60636,0.20481) -- (0.61038,0.20789) -- (0.61439,0.21100) -- (0.61841,0.21414) -- (0.62242,0.21732) -- (0.62644,0.22053) -- (0.63045,0.22377) -- (0.63447,0.22705) -- (0.63849,0.23037) -- (0.64250,0.23371) -- (0.64652,0.23710) -- (0.65053,0.24052) -- (0.65455,0.24398) -- (0.65856,0.24747) -- (0.66258,0.25101) -- (0.66660,0.25458) -- (0.67061,0.25819) -- (0.67463,0.26184) -- (0.67864,0.26553) -- (0.68266,0.26926) -- (0.68667,0.27303) -- (0.69069,0.27685) -- (0.69470,0.28070) -- (0.69872,0.28460) -- (0.70274,0.28855) -- (0.70675,0.29254) -- (0.71077,0.29657) -- (0.71478,0.30065) -- (0.71880,0.30478) -- (0.72281,0.30896) -- (0.72683,0.31318) -- (0.73085,0.31746) -- (0.73486,0.32178) -- (0.73888,0.32616) -- (0.74289,0.33059) -- (0.74691,0.33507) -- (0.75092,0.33961) -- (0.75494,0.34421) -- (0.75896,0.34886) -- (0.76297,0.35357) -- (0.76699,0.35834) -- (0.77100,0.36317) -- (0.77502,0.36806) -- (0.77903,0.37302) -- (0.78305,0.37804) -- (0.78706,0.38313) -- (0.79108,0.38829) -- (0.79510,0.39352) -- (0.79911,0.39882) -- (0.80313,0.40419) -- (0.80714,0.40965) -- (0.81116,0.41518) -- (0.81517,0.42079) -- (0.81919,0.42648) -- (0.82321,0.43226) -- (0.82722,0.43813) -- (0.83124,0.44408) -- (0.83525,0.45014) -- (0.83927,0.45629) -- (0.84328,0.46253) -- (0.84730,0.46889) -- (0.85132,0.47535) -- (0.85533,0.48192) -- (0.85935,0.48861) -- (0.86336,0.49542) -- (0.86738,0.50235) -- (0.87139,0.50942) -- (0.87541,0.51662) -- (0.87943,0.52397) -- (0.88344,0.53146) -- (0.88746,0.53911) -- (0.89147,0.54693) -- (0.89549,0.55492) -- (0.89950,0.56309) -- (0.90352,0.57146) -- (0.90753,0.58003) -- (0.91155,0.58882) -- (0.91557,0.59784) -- (0.91958,0.60711) -- (0.92360,0.61664) -- (0.92761,0.62646) -- (0.93163,0.63660) -- (0.93564,0.64706) -- (0.93966,0.65790) -- (0.94368,0.66914) -- (0.94769,0.68082) -- (0.95171,0.69300) -- (0.95572,0.70574) -- (0.95974,0.71912) -- (0.96375,0.73322) -- (0.96777,0.74818) -- (0.97179,0.76415) -- (0.97580,0.78136) -- (0.97982,0.80012) -- (0.98383,0.82093) -- (0.98785,0.84460) -- (0.99186,0.87273) -- (0.99588,0.90936) -- (1,1); 
    		\node at (0.45,0.11) [rectangle,draw=white,rotate=10,fill=white] {\textcolor{myblue}{$n = 2$}};
    
    		\draw[color=myblue2] (0.50195,0) -- (0.50597,0.00005) -- (0.50998,0.00015) -- (0.51400,0.00029) -- (0.51802,0.00049) -- (0.52203,0.00073) -- (0.52605,0.00103) -- (0.53006,0.00137) -- (0.53408,0.00177) -- (0.53809,0.00222) -- (0.54211,0.00272) -- (0.54613,0.00327) -- (0.55014,0.00387) -- (0.55416,0.00453) -- (0.55817,0.00523) -- (0.56219,0.00600) -- (0.56620,0.00681) -- (0.57022,0.00768) -- (0.57423,0.00861) -- (0.57825,0.00959) -- (0.58227,0.01063) -- (0.58628,0.01172) -- (0.59030,0.01287) -- (0.59431,0.01408) -- (0.59833,0.01534) -- (0.60234,0.01667) -- (0.60636,0.01805) -- (0.61038,0.01949) -- (0.61439,0.02099) -- (0.61841,0.02255) -- (0.62242,0.02418) -- (0.62644,0.02586) -- (0.63045,0.02761) -- (0.63447,0.02942) -- (0.63849,0.03130) -- (0.64250,0.03324) -- (0.64652,0.03525) -- (0.65053,0.03732) -- (0.65455,0.03946) -- (0.65856,0.04167) -- (0.66258,0.04394) -- (0.66660,0.04629) -- (0.67061,0.04870) -- (0.67463,0.05119) -- (0.67864,0.05375) -- (0.68266,0.05639) -- (0.68667,0.05909) -- (0.69069,0.06188) -- (0.69470,0.06474) -- (0.69872,0.06768) -- (0.70274,0.07069) -- (0.70675,0.07379) -- (0.71077,0.07697) -- (0.71478,0.08023) -- (0.71880,0.08357) -- (0.72281,0.08700) -- (0.72683,0.09052) -- (0.73085,0.09412) -- (0.73486,0.09781) -- (0.73888,0.10160) -- (0.74289,0.10548) -- (0.74691,0.10945) -- (0.75092,0.11352) -- (0.75494,0.11769) -- (0.75896,0.12195) -- (0.76297,0.12633) -- (0.76699,0.13080) -- (0.77100,0.13538) -- (0.77502,0.14007) -- (0.77903,0.14487) -- (0.78305,0.14979) -- (0.78706,0.15482) -- (0.79108,0.15997) -- (0.79510,0.16525) -- (0.79911,0.17064) -- (0.80313,0.17617) -- (0.80714,0.18183) -- (0.81116,0.18762) -- (0.81517,0.19355) -- (0.81919,0.19963) -- (0.82321,0.20585) -- (0.82722,0.21223) -- (0.83124,0.21875) -- (0.83525,0.22544) -- (0.83927,0.23230) -- (0.84328,0.23933) -- (0.84730,0.24653) -- (0.85132,0.25392) -- (0.85533,0.26149) -- (0.85935,0.26927) -- (0.86336,0.27725) -- (0.86738,0.28544) -- (0.87139,0.29385) -- (0.87541,0.30250) -- (0.87943,0.31138) -- (0.88344,0.32052) -- (0.88746,0.32992) -- (0.89147,0.33960) -- (0.89549,0.34957) -- (0.89950,0.35985) -- (0.90352,0.37045) -- (0.90753,0.38140) -- (0.91155,0.39271) -- (0.91557,0.40440) -- (0.91958,0.41652) -- (0.92360,0.42907) -- (0.92761,0.44211) -- (0.93163,0.45566) -- (0.93564,0.46976) -- (0.93966,0.48448) -- (0.94368,0.49987) -- (0.94769,0.51599) -- (0.95171,0.53294) -- (0.95572,0.55081) -- (0.95974,0.56973) -- (0.96375,0.58986) -- (0.96777,0.61139) -- (0.97179,0.63459) -- (0.97580,0.65982) -- (0.97982,0.68761) -- (0.98383,0.71876) -- (0.98785,0.75460) -- (0.99186,0.79773) -- (0.99588,0.85479) -- (1,1); 
    		\node at (0.73,0.1) [rectangle,draw=white,rotate=16,fill=white] {\textcolor{myblue2}{$n = 3$}};
    
    		\draw[color=myblue3] (0.66660,0) -- (0.67061,0.00004) -- (0.67463,0.00019) -- (0.67864,0.00043) -- (0.68266,0.00078) -- (0.68667,0.00122) -- (0.69069,0.00177) -- (0.69470,0.00242) -- (0.69872,0.00318) -- (0.70274,0.00405) -- (0.70675,0.00503) -- (0.71077,0.00611) -- (0.71478,0.00731) -- (0.71880,0.00862) -- (0.72281,0.01005) -- (0.72683,0.01160) -- (0.73085,0.01326) -- (0.73486,0.01505) -- (0.73888,0.01696) -- (0.74289,0.01900) -- (0.74691,0.02116) -- (0.75092,0.02346) -- (0.75494,0.02588) -- (0.75896,0.02844) -- (0.76297,0.03114) -- (0.76699,0.03399) -- (0.77100,0.03697) -- (0.77502,0.04010) -- (0.77903,0.04338) -- (0.78305,0.04681) -- (0.78706,0.05040) -- (0.79108,0.05415) -- (0.79510,0.05806) -- (0.79911,0.06214) -- (0.80313,0.06639) -- (0.80714,0.07082) -- (0.81116,0.07543) -- (0.81517,0.08022) -- (0.81919,0.08521) -- (0.82321,0.09039) -- (0.82722,0.09577) -- (0.83124,0.10136) -- (0.83525,0.10716) -- (0.83927,0.11318) -- (0.84328,0.11943) -- (0.84730,0.12592) -- (0.85132,0.13265) -- (0.85533,0.13963) -- (0.85935,0.14687) -- (0.86336,0.15438) -- (0.86738,0.16218) -- (0.87139,0.17026) -- (0.87541,0.17866) -- (0.87943,0.18737) -- (0.88344,0.19641) -- (0.88746,0.20581) -- (0.89147,0.21557) -- (0.89549,0.22571) -- (0.89950,0.23626) -- (0.90352,0.24724) -- (0.90753,0.25866) -- (0.91155,0.27057) -- (0.91557,0.28299) -- (0.91958,0.29595) -- (0.92360,0.30949) -- (0.92761,0.32366) -- (0.93163,0.33851) -- (0.93564,0.35409) -- (0.93966,0.37047) -- (0.94368,0.38772) -- (0.94769,0.40594) -- (0.95171,0.42524) -- (0.95572,0.44574) -- (0.95974,0.46761) -- (0.96375,0.49106) -- (0.96777,0.51634) -- (0.97179,0.54380) -- (0.97580,0.57391) -- (0.97982,0.60735) -- (0.98383,0.64516) -- (0.98785,0.68908) -- (0.99186,0.74250) -- (0.99588,0.81401) -- (1,1); 
    		\node at (0.82,0.1) [rectangle,draw=white,rotate=26,fill=white] {\textcolor{myblue3}{$n = 4$}};
    
    		\draw[color=myblue4] (0.75092,0) -- (0.75494,0.00012) -- (0.75896,0.00040) -- (0.76297,0.00085) -- (0.76699,0.00148) -- (0.77100,0.00227) -- (0.77502,0.00325) -- (0.77903,0.00441) -- (0.78305,0.00575) -- (0.78706,0.00729) -- (0.79108,0.00902) -- (0.79510,0.01094) -- (0.79911,0.01307) -- (0.80313,0.01541) -- (0.80714,0.01797) -- (0.81116,0.02074) -- (0.81517,0.02373) -- (0.81919,0.02696) -- (0.82321,0.03043) -- (0.82722,0.03413) -- (0.83124,0.03809) -- (0.83525,0.04231) -- (0.83927,0.04680) -- (0.84328,0.05156) -- (0.84730,0.05660) -- (0.85132,0.06194) -- (0.85533,0.06758) -- (0.85935,0.07354) -- (0.86336,0.07982) -- (0.86738,0.08645) -- (0.87139,0.09342) -- (0.87541,0.10077) -- (0.87943,0.10850) -- (0.88344,0.11663) -- (0.88746,0.12518) -- (0.89147,0.13418) -- (0.89549,0.14363) -- (0.89950,0.15357) -- (0.90352,0.16403) -- (0.90753,0.17502) -- (0.91155,0.18660) -- (0.91557,0.19878) -- (0.91958,0.21162) -- (0.92360,0.22516) -- (0.92761,0.23945) -- (0.93163,0.25455) -- (0.93564,0.27053) -- (0.93966,0.28747) -- (0.94368,0.30546) -- (0.94769,0.32460) -- (0.95171,0.34504) -- (0.95572,0.36691) -- (0.95974,0.39043) -- (0.96375,0.41583) -- (0.96777,0.44342) -- (0.97179,0.47362) -- (0.97580,0.50698) -- (0.97982,0.54433) -- (0.98383,0.58689) -- (0.98785,0.63676) -- (0.99186,0.69795) -- (0.99588,0.78073) -- (1,1); 
    
    		\draw[color=myblue5] (0.79911,0) -- (0.80313,0.00007) -- (0.80714,0.00038) -- (0.81116,0.00095) -- (0.81517,0.00178) -- (0.81919,0.00287) -- (0.82321,0.00424) -- (0.82722,0.00589) -- (0.83124,0.00782) -- (0.83525,0.01006) -- (0.83927,0.01260) -- (0.84328,0.01546) -- (0.84730,0.01865) -- (0.85132,0.02217) -- (0.85533,0.02605) -- (0.85935,0.03029) -- (0.86336,0.03490) -- (0.86738,0.03991) -- (0.87139,0.04532) -- (0.87541,0.05115) -- (0.87943,0.05743) -- (0.88344,0.06416) -- (0.88746,0.07138) -- (0.89147,0.07911) -- (0.89549,0.08736) -- (0.89950,0.09618) -- (0.90352,0.10558) -- (0.90753,0.11561) -- (0.91155,0.12629) -- (0.91557,0.13768) -- (0.91958,0.14982) -- (0.92360,0.16276) -- (0.92761,0.17656) -- (0.93163,0.19129) -- (0.93564,0.20703) -- (0.93966,0.22386) -- (0.94368,0.24189) -- (0.94769,0.26125) -- (0.95171,0.28208) -- (0.95572,0.30457) -- (0.95974,0.32893) -- (0.96375,0.35544) -- (0.96777,0.38446) -- (0.97179,0.41645) -- (0.97580,0.45207) -- (0.97982,0.49223) -- (0.98383,0.53836) -- (0.98785,0.59282) -- (0.99186,0.66022) -- (0.99588,0.75226) -- (1,1); 
    
    		\draw[color=myblue6] (0.83124,0) -- (0.83525,0.00003) -- (0.83927,0.00037) -- (0.84328,0.00106) -- (0.84730,0.00212) -- (0.85132,0.00356) -- (0.85533,0.00538) -- (0.85935,0.00762) -- (0.86336,0.01027) -- (0.86738,0.01336) -- (0.87139,0.01690) -- (0.87541,0.02092) -- (0.87943,0.02543) -- (0.88344,0.03045) -- (0.88746,0.03601) -- (0.89147,0.04214) -- (0.89549,0.04885) -- (0.89950,0.05619) -- (0.90352,0.06419) -- (0.90753,0.07289) -- (0.91155,0.08232) -- (0.91557,0.09253) -- (0.91958,0.10358) -- (0.92360,0.11552) -- (0.92761,0.12842) -- (0.93163,0.14235) -- (0.93564,0.15741) -- (0.93966,0.17368) -- (0.94368,0.19130) -- (0.94769,0.21039) -- (0.95171,0.23112) -- (0.95572,0.25369) -- (0.95974,0.27834) -- (0.96375,0.30539) -- (0.96777,0.33523) -- (0.97179,0.36838) -- (0.97580,0.40555) -- (0.97982,0.44779) -- (0.98383,0.49665) -- (0.98785,0.55478) -- (0.99186,0.62728) -- (0.99588,0.72717) -- (1,1); 
    
    		\draw[color=myblue7] (0.85533,0) -- (0.85935,0.00006) -- (0.86336,0.00055) -- (0.86738,0.00151) -- (0.87139,0.00297) -- (0.87541,0.00495) -- (0.87943,0.00746) -- (0.88344,0.01055) -- (0.88746,0.01422) -- (0.89147,0.01850) -- (0.89549,0.02344) -- (0.89950,0.02906) -- (0.90352,0.03540) -- (0.90753,0.04249) -- (0.91155,0.05040) -- (0.91557,0.05916) -- (0.91958,0.06883) -- (0.92360,0.07949) -- (0.92761,0.09119) -- (0.93163,0.10402) -- (0.93564,0.11807) -- (0.93966,0.13347) -- (0.94368,0.15032) -- (0.94769,0.16879) -- (0.95171,0.18905) -- (0.95572,0.21132) -- (0.95974,0.23587) -- (0.96375,0.26303) -- (0.96777,0.29325) -- (0.97179,0.32707) -- (0.97580,0.36530) -- (0.97982,0.40905) -- (0.98383,0.46004) -- (0.98785,0.52115) -- (0.99186,0.59794) -- (0.99588,0.70463) -- (1,1); 
    
    		\draw[color=myblue8] (0.87541,0) -- (0.87943,0.00035) -- (0.88344,0.00132) -- (0.88746,0.00292) -- (0.89147,0.00519) -- (0.89549,0.00816) -- (0.89950,0.01186) -- (0.90352,0.01635) -- (0.90753,0.02165) -- (0.91155,0.02783) -- (0.91557,0.03493) -- (0.91958,0.04301) -- (0.92360,0.05215) -- (0.92761,0.06242) -- (0.93163,0.07392) -- (0.93564,0.08673) -- (0.93966,0.10099) -- (0.94368,0.11682) -- (0.94769,0.13440) -- (0.95171,0.15391) -- (0.95572,0.17558) -- (0.95974,0.19972) -- (0.96375,0.22667) -- (0.96777,0.25691) -- (0.97179,0.29105) -- (0.97580,0.32994) -- (0.97982,0.37477) -- (0.98383,0.42741) -- (0.98785,0.49095) -- (0.99186,0.57140) -- (0.99588,0.68408) -- (1,1); 
    
    		\draw[color=myblue9] (0.88746,0) -- (0.89147,0.00015) -- (0.89549,0.00100) -- (0.89950,0.00265) -- (0.90352,0.00512) -- (0.90753,0.00847) -- (0.91155,0.01276) -- (0.91557,0.01803) -- (0.91958,0.02435) -- (0.92360,0.03181) -- (0.92761,0.04047) -- (0.93163,0.05044) -- (0.93564,0.06182) -- (0.93966,0.07474) -- (0.94368,0.08934) -- (0.94769,0.10581) -- (0.95171,0.12434) -- (0.95572,0.14519) -- (0.95974,0.16866) -- (0.96375,0.19514) -- (0.96777,0.22513) -- (0.97179,0.25929) -- (0.97580,0.29851) -- (0.97982,0.34409) -- (0.98383,0.39800) -- (0.98785,0.46355) -- (0.99186,0.54714) -- (0.99588,0.66514) -- (1,1); 
      
    		\draw[thick,-to] (-0.02,0) -- (1.02,0) node[anchor=west]{$\gamma$};
    		\draw[thick,-to] (0,-0.05) -- (0,1.05) node[anchor=south]{SDP value};
    	\end{tikzpicture}
    	\caption{Plot of the relationship between $\gamma$ and the optimal SDP value when performing state exclusion on the set of states described by Example~\ref{exam:equi_angle}, for $2 \leq n \leq 10$. 
        This value equals $0$ if and only if the set is antidistinguishable, which happens exactly when $\gamma \leq (n-2)/(n-1)$. Dividing the SDP value by $n$ yields the optimal error probability in the state exclusion task, when all of the states are chosen uniformly at random.
        }\label{fig:santa_graph}
    \end{figure}
\end{example}

\section{Necessary conditions for antidistinguishability}\label{sec:upper_bounds}

While antidistinguishability of a set can be checked via semidefinite programming, it is useful to have necessary and/or sufficient conditions for antidistinguishability that are even easier to make use of (e.g., conditions that rely only on elementary linear algebra, or on quantities that have a natural physical interpretation). In this section, we present a pair of necessary conditions for antidistinguishability that involve just inequalities of the inner products of the pure states. In particular, if these inner products are sufficiently large then the set cannot be antidistinguishable:

\begin{theorem}\label{thm:non-antidist-threshold}
    Let $n \geq 2$ be an integer and let $S = \{\ket{\psi_0}, \ket{\psi_1}, \ldots, \ket{\psi_{n-1}} \}$. If
    \begin{equation}\label{eq:inner-product-non-antidist-threshold}
        \sum_{i \neq j = 0}^{n-1} \big|\ip{\psi_i}{\psi_j}\big| > n(n-2)
    \end{equation}
    then $S$ is not antidistinguishable.
\end{theorem}

We note that the above theorem also follows from the work  in~\cite{bandyopadhyay2014conclusive} which considered the more general task of quantum state exclusion of mixed states. We give an alternative proof using Theorem~\ref{thm:n_1_incoh} in just a couple of lines, to demonstrate how simple antidistinguishability of pure states is to work with via our $(n-1)$-incoherence machinery.
 
\begin{proof}[Proof of Theorem~\ref{thm:non-antidist-threshold}]
    This result follows via an argument that is similar to the one used in the latter half of Example~\ref{exam:equi_angle}. Define
    \[
        Y = (n-1)I - E,
    \]
    where, for all $0 \leq i,j \leq n-1$ (even if $i = j$), the $(i,j)$-entry of $E$ is the complex number with modulus $1$ and phase equal to that of $\ip{\psi_i}{\psi_j}$. Since each $(n-1) \times (n-1)$ principal submatrix of $Y$ is diagonally dominant, $Y$ is $(n-1)$-locally PSD. However,
    \begin{align*}
        \tr(YG) & = \tr\Big( \big((n-1)I - E\big)G \Big) = n(n-2) - \sum_{i\neq j=0}^{n-1}\big|\ip{\psi_i}{\psi_j}\big|,
    \end{align*}
    which is strictly negative whenever Inequality~\eqref{eq:inner-product-non-antidist-threshold} holds. It follows that $G$ is not $(n-1)$-incoherent, so Theorem~\ref{thm:n_1_incoh} tells us that $S$ is not antidistinguishable.
\end{proof}

By noting that there are $n(n-1)$ terms in the sum~\eqref{eq:inner-product-non-antidist-threshold}, the above theorem immediately implies the following special case:

\begin{corollary}\label{cor:large_ip_not_anti}
    Let $n \geq 2$ be an integer and let $S = \{\ket{\psi_0}, \ket{\psi_1}, \ldots, \ket{\psi_{n-1}} \}$. If
    \begin{equation}\label{eq:inner-product-non-antidist-individual}
        \big|\ip{\psi_i}{\psi_j}\big| > \frac{n-2}{n-1} \quad \text{for all} \quad 0 \leq i \neq j \leq n-1
    \end{equation}
    then $S$ is not antidistinguishable.
\end{corollary}

We note that Example~\ref{exam:equi_angle} demonstrates that the inequalities described by Theorem~\ref{thm:non-antidist-threshold} and Corollary~\ref{cor:large_ip_not_anti} are both tight: for all $n$, there is an antidistinguishable set with $\big|\ip{\psi_i}{\psi_j}\big| = (n-2)/(n-1)$ and thus $\sum_{i \neq j = 0}^{n-1} \big|\ip{\psi_i}{\psi_j}\big| = n(n-2)$. In the $n = 3$ case, the trine states from Example~\ref{exam:trine} also demonstrate tightness of these bounds, as they are antidistinguishable with $\big|\ip{\psi_i}{\psi_j}\big| = (n-2)/(n-1) = 1/2$ for all $i \neq j$.

\section{Sufficient conditions for antidistinguishability}\label{sec:lower_bounds}

We now present some sufficient conditions for antidistinguishability that are simpler to use than any of the semidefinite programs we have described. Much like the results of Section~\ref{sec:upper_bounds} showed that if the states' inner products are sufficiently large, then the set cannot be antidistinguishable, in this section we show that if the inner products are sufficiently small, then the set \emph{must} be antidistinguishable. In particular, one of these sufficient conditions (Corollary~\ref{cor:antidist_by_IP}) can be thought of as a ``corrected version'' of the recently-disproved conjecture~\cite{russo2023inner} from \cite{havlivcek2020simple}. 
For completeness, their conjecture was for an integer $n \geq 2$ and $S = \{\ket{\psi_0}, \ket{\psi_1}, \ldots, \ket{\psi_{n-1}} \}$ if
\begin{equation}\label{eq:inner-product-non-antidist-individual-incorrect}
    \big|\ip{\psi_i}{\psi_j}\big| \leq \frac{n-2}{n-1} \quad \text{for all} \quad 0 \leq i \neq j \leq n-1
\end{equation}
then $S$ is antidistinguishable. 

As our first step towards correcting this conjecture (in particular, placing a correct quantity on the right-hand-side of Inequality~\eqref{eq:inner-product-non-antidist-individual-incorrect}), we present a sufficient condition for antidistinguishability in terms of the eigenvalues of the set's Gram matrix. Remarkably, this sufficient condition is also necessary for circulant sets:
 
\begin{theorem}\label{thm:antidist_by_eigs}
    Let $n \geq 2$ be an integer and let $G \in \Pos(\C^n)$ be the Gram matrix of a set $S$ of $n$ pure states, and let $\lambda_0 \geq \lambda_1 \geq \cdots \geq \lambda_{n-1}$ be the eigenvalues of $G$. If
    \begin{align}\label{eq:eig_ineq_for_antidist}
        \sqrt{\lambda_0} \leq \sum_{j=1}^{n-1} \sqrt{\lambda_j}
    \end{align}
    then $S$ is antidistinguishable. Furthermore, if $G$ is circulant, then Inequality~\eqref{eq:eig_ineq_for_antidist} is necessary and sufficient for the antidistinguishability of $S$.
\end{theorem}

\begin{proof}
    Define $q := \sum_{j=1}^{n-1} \sqrt{\lambda_j}$ and suppose that $\sqrt{\lambda_0} \leq q$. Our goal is to show that $S$ is antidistinguishable. It was shown in \cite[Theorem~8]{JMPP22} that if there exists a real matrix $\Lambda \in \Pos(\R^n)$ such that
    \begin{align}\label{eq:abs_d1_incohtest}
        \lambda_0 & = -\Lambda_{0,0} - \sum_{i=1}^{n-1} \left(\Lambda_{0,i} + \Lambda_{i,0}\right) \quad \text{and} \quad \lambda_j = \Lambda_{j,j} \ \ \ \text{for ${} \ 1 \leq j \leq n-1$}
    \end{align}
    then $G$ is $(n-1)$-incoherent, so (by Theorem~\ref{thm:n_1_incoh}) $S$ is antidistinguishable. It thus suffices to find such a $\Lambda$.
    
    To this end, define a vector $\mathbf{v} \in \R^n$ by
    \[
        v_0 = -q - \sqrt{q^2 - \lambda_0} \quad \text{and} \quad v_j = \sqrt{\lambda_j} \ \ \ \text{for ${} \ 1 \leq j \leq n-1$}
    \]
    (the hypothesis $\sqrt{\lambda_0} \leq q$ was used here to ensure that $v_0$ is real). It is then straightforward to check that the positive semidefinite matrix $\Lambda = \mathbf{v}\mathbf{v}^T$ satisfies Equation~\eqref{eq:abs_d1_incohtest}, which completes the proof that Inequality~\eqref{eq:eig_ineq_for_antidist} implies antidistinguishability of $S$.

    For the ``furthermore'' statement, suppose that $G$ is circulant. Our goal is to show that $S$ being antidistinguishable is equivalent to Inequality~\eqref{eq:eig_ineq_for_antidist} holding. To this end, recall from Theorem~\ref{thm:n_1_incoh} that $S$ is antidistinguishable if and only if $G$ is $(n-1)$-incoherent, which (by Lemma~\ref{lem:circulant_n1_inc}) is equivalent to
    \begin{align}\label{ineq:n1_incoh_circ_dual}
        \tr(GY) \geq 0
    \end{align}
    for all circulant $(n-1)$-locally PSD matrices $Y$. Well, Proposition~\ref{prop:circulant}(c) tells us that a matrix $Y$ is circulant if and only if $Y = F\mathrm{diag}(\mathbf{d})F^*$, where $\mathbf{d}$ is the vector of eigenvalues of $Y$. Furthermore, it was shown in~\cite[Theorem~2]{JMPP22} that a circulant matrix $Y$ is $(n-1)$-locally PSD if and only if $S_{k}(\mathbf{d}) \geq 0$ for all $1 \leq k \leq n-1$, where $S_{k}$ is the $k$-th elementary symmetric polynomial
    \[
        S_k(\mathbf{d}) := \sum_{0 \leq j_1 < \cdots < j_k < n} d_{j_1}d_{j_2}\cdots d_{j_k}.
    \]
    
    Since $G$ is circulant, we can write $G = F\mathrm{diag}(\bm{\lambda})F^*$ for some vector $\bm{\lambda} = (\lambda_0, \ldots, \lambda_{n-1})$ whose entries are the eigenvalues of $G$. It follows that Inequality~\eqref{ineq:n1_incoh_circ_dual} is equivalent to
    \[
        0 \leq \tr(GY) = \tr\big((F\mathrm{diag}(\bm{\lambda})F^*)(F\mathrm{diag}(\mathbf{d})F^*)\big) = \tr\big(\mathrm{diag}(\bm{\lambda})\mathrm{diag}(\mathbf{d})\big) = \bm{\lambda} \cdot \mathbf{d}.
    \]
    In other words, $S$ is antidistinguishable if and only if $\bm{\lambda}$ is in the dual cone of the set of vectors $\mathbf{d}$ satisfying $S_{k}(\mathbf{d}) \geq 0$ for all $1 \leq k \leq n-1$. This dual cone was characterized in \cite[Proposition~4.2]{Zin08}, and $\bm{\lambda}$ being in this dual cone implies the existence of $\Lambda \in \Pos(\R^n)$ satisfying Equation~\eqref{eq:abs_d1_incohtest}. Since $\Lambda$ is positive semidefinite, so are all of its $2 \times 2$ principal submatrices, so $\Lambda_{0,0}\Lambda_{j,j} \geq \Lambda_{0,j}^2$ for all $1 \leq j \leq n-1$. Substituting this into Equation~\eqref{eq:abs_d1_incohtest} gives
    \[
      \Lambda_{0,0} = -\lambda_0 - 2\sum_{j=1}^{n-1}\Lambda_{0,j} \leq -\lambda_0 + 2\sqrt{\Lambda_{0,0}}\sum_{j=1}^{n-1}\sqrt{\lambda_{j}}.
    \]
    This is a quadratic inequality in $\sqrt{\Lambda_{0,0}}$, which (via the discriminant of the quadratic formula) has a real solution if and only if
    \[
      \left(\sum_{j=1}^{n-1}\sqrt{\lambda_{j}}\right)^2 - \lambda_0 \geq 0,
    \]
    which implies $\sqrt{\lambda_0} \leq \sum_{j=1}^{n-1} \sqrt{\lambda_j}$.
\end{proof}

\begin{remark}\label{rem:abs_n1_no_sdp}
    The proof of Theorem~\ref{thm:antidist_by_eigs} shows something that was overlooked in \cite{Zin08,JMPP22}: the existence of $\Lambda \in \Pos(\R^n)$ satisfying the constraints~\eqref{eq:abs_d1_incohtest} can be determined without semidefinite programming. Inequality~\eqref{eq:eig_ineq_for_antidist} is both necessary and sufficient for the existence of such a $\Lambda$.
\end{remark}

Our first corollary of Theorem~\ref{thm:antidist_by_eigs} gives a sufficient condition for antidistinguishability in terms of the Frobenius norm $\|G\|_{\textup{F}}$ of $G$, which is slightly easier to compute than its eigenvalues.

\begin{corollary}\label{cor:antidist_by_frob}
    Let $n \geq 2$ be an integer and let $G$ be the Gram matrix of a set $S$ of $n$ pure states. If
    \begin{align*}
        \|G\|_{\textup{F}} \leq \frac{n}{\sqrt{2}}
    \end{align*}
    then $S$ is antidistinguishable.
\end{corollary}

It is perhaps worth noting that for all Gram matrices $G$ we have $\sqrt{n} \leq \|G\|_{\textup{F}} \leq n$, with the lower bound being saturated when the states in the set are mutually orthogonal and the upper bound being saturated when the states in the set are all equal to each other. Corollary~\ref{cor:antidist_by_frob} thus says that if the states in a set are ``close enough'' to being mutually orthogonal, they must be antidistinguishable.

\begin{proof}[Proof of Corollary~\ref{cor:antidist_by_frob}]
    Define $x_j = \sqrt{\lambda_j/n}$ for all $0 \leq j \leq n-1$, so that the conditions $\tr(G) = n$ and $\|G\|_{\textup{F}} \leq n/\sqrt{2}$ are equivalent to
    \begin{align}\label{eq:frob_const1}
        \sum_{j=0}^{n-1} x_j^2 = 1 \quad \text{and} \quad \sum_{j=0}^{n-1} x_j^4 \leq \frac{1}{2},
    \end{align}
    respectively. If we can show that these conditions imply $\sqrt{\lambda_0} \leq \sum_{j=1}^{n-1} \sqrt{\lambda_j}$ (i.e., $x_0 \leq \sum_{j=1}^{n-1} x_j$) then Theorem~\ref{thm:antidist_by_eigs} will imply the present corollary.
    
    To this end, we note that the constraints~\eqref{eq:frob_const1} imply $x_0^2 = 1 - \sum_{j=1}^{n-1} x_j^2$ and thus
    \begin{align*}
        \left(1 - \sum_{j=1}^{n-1} x_j^2\right)^2 + \sum_{j=1}^{n-1} x_j^4 \leq \frac{1}{2}.
    \end{align*}
    Multiplying through by $2$ and rearranging slightly shows that the above inequality is equivalent to
    \[
        \left(1 - 2\sum_{j=1}^{n-1} x_j^2\right)^2 \leq 4\sum_{\substack{i,j=1\\ i > j}}^{n-1}x_{i}^2x_{j}^2.
    \]
    Using the fact that $1 = \sum_{j=0}^{n-1} x_j^2$ on the left-hand-side, and then square-rooting both sides (noting that the right-hand-side is non-negative, so the direction of the inequality is preserved) then shows that
    \begin{align}\label{ineq:frob_proof_string}
        x_0^2 - \sum_{j=1}^{n-1} x_j^2 \leq 2\sqrt{\sum_{\substack{i,j=1\\
        i > j}}^{n-1}x_{i}^2x_{j}^2} \leq 2\sum_{\substack{i,j=1\\
        i > j}}^{n-1}x_{i}x_{j},
    \end{align}
    where the second inequality follows since the $2$-norm is at most the $1$-norm. In particular, the outermost inequality in~\eqref{ineq:frob_proof_string} can be rearranged as
    \begin{align*}
        x_0^2 & \leq \sum_{j=1}^{n-1} x_j^2 + 2\sum_{\substack{i,j=1\\
        i > j}}^{n-1}x_{i}x_{j},
    \end{align*}
    which can be factored as
    \begin{align*}
        x_0^2 & \leq \left(\sum_{j=1}^{n-1} x_j\right)^2.
    \end{align*}
    This implies $x_0 \leq \sum_{j=1}^{n-1} x_j$, as desired.
\end{proof}

Our final sufficient condition for antidistinguishability arises simply by noting that if each off-diagonal entry of a Gram matrix $G$ has $|g_{i,j}| \leq \sqrt{(n-2)/(2n-2)}$, then $\|G\|_{\textup{F}} \leq n/\sqrt{2}$, so Corollary~\ref{cor:antidist_by_frob} tells us that the set is antidistinguishable. In other words, we have the following corrected version of the conjecture from \cite{havlivcek2020simple}:

\begin{corollary}\label{cor:antidist_by_IP}
    Let $n \geq 2$ be an integer and let $S = \{\ket{\psi_0},\ket{\psi_1},\ldots,\ket{\psi_{n-1}}\}$. If
    \begin{align}\label{eq:correct_inequality}
        \big|\braket{\psi_i}{\psi_j}\big| \leq \frac{1}{\sqrt{2}}\sqrt{\frac{n-2}{n-1}} \quad \text{for all} \quad 0 \leq i \neq j \leq n-1
    \end{align}
    then $S$ is antidistinguishable.
\end{corollary}

\section{Tightness of these bounds}\label{sec:tightness}

We already noted that the bounds of Section~\ref{sec:upper_bounds} are tight, as demonstrated by Example~\ref{exam:equi_angle} (which is circulant). The bound of Theorem~\ref{thm:antidist_by_eigs} is also tight, as shown by the fact that it is both necessary and sufficient for circulant matrices. Corollary~\ref{cor:antidist_by_frob} can also be seen to be tight via circulant matrices: if $\varepsilon > 0$ is small and $G$ is a circulant Gram matrix with eigenvalues $\lambda_0 = n/2+\varepsilon$, $\lambda_1 = n/2-\varepsilon$, and $\lambda_j = 0$ for $j \geq 2$, then
\[
    \sum_{j=1}^{n-1}\sqrt{\lambda_j} = \sqrt{\frac{n}{2} - \varepsilon} < \sqrt{\frac{n}{2} + \varepsilon} = \sqrt{\lambda_0},
\]
so any set of pure states with Gram matrix $G$ must not be antidistinguishable, by Theorem~\ref{thm:antidist_by_eigs}. However, in this case, we also have
\[
    \|G\|_{\textup{F}} = \sqrt{\left(\frac{n}{2} + \varepsilon\right)^2 + \left(\frac{n}{2} - \varepsilon\right)^2} = \sqrt{\frac{n^2}{2} + 2\varepsilon^2} \leq \frac{n}{\sqrt{2}} + \sqrt{2} \varepsilon,
\]
demonstrating that the quantity $n/\sqrt{2}$ in Corollary~\ref{cor:antidist_by_frob} cannot be increased at all.

The only question remaining is whether or not the bound established by Corollary~\ref{cor:antidist_by_IP} is also tight. It is trivially tight when $n = 2$ or $n = 3$ since it matches the corresponding necessary condition provided by Corollary~\ref{cor:large_ip_not_anti}. The situation is less clear when $n \geq 4$, since when $n = 4$, for example, we have the following situation:

\begin{itemize}
    \item If $|\braket{\psi_i}{\psi_j}| > 2/3 \approx 0.6667$ for all $i \neq j$ then the set is not antidistinguishable (Corollary~\ref{cor:large_ip_not_anti}).
    
    \item If $|\braket{\psi_i}{\psi_j}| \leq 2/3$ for all $i \neq j$ then the set \emph{may} be antidistinguishable (Example~\ref{exam:equi_angle}), so Corollary~\ref{cor:large_ip_not_anti} is tight.
    
    \item If $|\braket{\psi_i}{\psi_j}| \leq 1/\sqrt{3} \approx 0.5774$ for all $i \neq j$ then the set is antidistinguishable (Corollary~\ref{cor:antidist_by_IP}).
    
    \item It is currently only known that the set \emph{may} not be antidistinguishable when $|\braket{\psi_i}{\psi_j}| > 0.6451$ \cite{russo2023inner}. Our next example improves this bound to $1/\sqrt{3}$, thus showing that Corollary~\ref{cor:antidist_by_IP} is tight, at least when $n = 4$:
\end{itemize}

\begin{example}\label{ex:d4_lb_tight}
    Let $c = 1/\sqrt{3}$ and consider the matrix
    \[
        G := \begin{bmatrix}
            1 & c & c & c \\
            c & 1 & ci & (1 + ci)/2 \\
            c & -ci & 1 & (1 - ci)/2 \\
            c & (1 - ci)/2 & (1 + ci)/2 & 1
        \end{bmatrix}.
    \]
    It is straightforward to check that $G$ is positive semidefinite and is thus the Gram matrix of some set $S$ of $n = 4$ pure states. Since $|g_{i,j}| = 1/\sqrt{3}$ for all $i \neq j$, Corollary~\ref{cor:antidist_by_IP} tells us that $S$ is antidistinguishable.
    
    Now let $\varepsilon > 0$ be small, let $\mathbf{v} = [1, (-\sqrt{3} + i)/2, (-\sqrt{3} - i)/2, 0]^T$ and $\mathbf{w} = [0,0,0,1]^T$, and define
    \[
        G_{\varepsilon} = \frac{1}{1 - 2\varepsilon}\big(G + \varepsilon (\mathbf{v}\mathbf{v}^* + \mathbf{w}\mathbf{w}^* - 3I)\big).
    \]
    A straightforward computation shows that $G_{\varepsilon}$ is positive semidefinite if $0 < \varepsilon < 1/10$ and its diagonal entries all equal $1$, so it is the Gram matrix of some set $S_\varepsilon$ of $n = 4$ pure states. Furthermore, $\lim_{\varepsilon \rightarrow 0^{+}}G_{\varepsilon} = G$, so the inner products of the members of $S_{\varepsilon}$ can be made to have modulus as close to $1/\sqrt{3}$ as we like by choosing $\varepsilon > 0$ sufficiently small.

    We claim that $G_{\varepsilon}$ is not $(n-1) = 3$-incoherent, so (by Theorem~\ref{thm:n_1_incoh}), $S_{\varepsilon}$ is not antidistinguishable. To verify this claim, we need to find a $3$-locally PSD matrix $X$ for which $\tr(XG_{\varepsilon}) < 0$. To this end, consider the matrices
    \begin{align*}
        Y & = \begin{bmatrix}
            2 & -\sqrt{3} - i & -\sqrt{3} + i & 0 \\
            -\sqrt{3} + i & 2 & 1 - \sqrt{3}i & 0 \\
            -\sqrt{3} - i & 1 + \sqrt{3}i & 2 & 0 \\
            0 & 0 & 0 & 0
        \end{bmatrix} \quad \text{and} \\
        Z & = \begin{bmatrix}
            0 & 1 + \sqrt{3}i & 1 - \sqrt{3}i & -2 \\
            1 - \sqrt{3}i & 0 & -\sqrt{3} + i & -\sqrt{3} - i \\
            1 + \sqrt{3}i & -\sqrt{3} - i & 0 & -\sqrt{3} + i \\
            -2 & -\sqrt{3} + i & -\sqrt{3} - i & 2\sqrt{3}(1 + 5\varepsilon)
        \end{bmatrix}.
    \end{align*}
    The following claims are all straightforward (albeit somewhat tedious) to verify:
    \begin{itemize}
        \item $\tr(YG_\varepsilon) = 0$ for all $0 < \varepsilon < 1/2$.
        
        \item $\tr(ZG_\varepsilon) = -(20\sqrt{3})\varepsilon^2/(1-2\varepsilon) < 0$ for all $0 < \varepsilon < 1/2$.
        
        \item The matrix $X = Y + \delta Z$ is $3$-locally PSD when $0 < \varepsilon < 1/2$ and $0 < \delta \leq 5\sqrt{3}\varepsilon/(1 + 5\varepsilon)$.

    \end{itemize}
    It follows that, for these choices of $\delta$ and $\varepsilon$, $X$ is a $3$-locally positive semidefinite matrix with
    \[
        \tr(XG_{\varepsilon}) = -(20\sqrt{3})\delta\varepsilon^2/(1-2\varepsilon) < 0,
    \]
    proving our claim.
\end{example}

We summarize the above example and related theorems in Figure~\ref{fig:antidist_ip_range}.

\begin{figure}[!htp]
    \begin{center}
    	\begin{tikzpicture}[scale=2.25]
    		\draw[gray!35] (1.8,0.25) -- (7.4,0.25);
    		\draw[gray!35] (1.8,0.5) -- (7.4,0.5);
    		\draw[gray!35] (1.8,0.75) -- (7.4,0.75);
    		\draw[gray!35] (1.8,1) -- (7.4,1);
    		
    		\draw[thick,->] (1.8,0) -- (7.4,0)node[anchor=west]{$n$};
    		\draw[thick,->] (1.8,0) -- (1.8,1.2)node[anchor=south]{$\iprod{\psi_i}{\psi_j}$};

    		\draw[thick] (2,-0.03)node[anchor=north]{$2$} -- (2,0.03);
    		\draw[thick] (3,-0.03)node[anchor=north]{$3$} -- (3,0.03);
    		\draw[thick] (4,-0.03)node[anchor=north]{$4$} -- (4,0.03);
    		\draw[thick] (5,-0.03)node[anchor=north]{$5$} -- (5,0.03);
    		\draw[thick] (6,-0.03)node[anchor=north]{$6$} -- (6,0.03);
    		\draw[thick] (7,-0.03)node[anchor=north]{$7$} -- (7,0.03);
    		
    		\draw[thick] (1.77,0)node[anchor=east]{$0$} -- (1.83,0);
    		\draw[thick] (1.77,0.25)node[anchor=east]{$0.25$} -- (1.83,0.25);
    		\draw[thick] (1.77,0.5)node[anchor=east]{$0.5$} -- (1.83,0.5);
    		\draw[thick] (1.77,0.75)node[anchor=east]{$0.75$} -- (1.83,0.75);
    		\draw[thick] (1.77,1)node[anchor=east]{$1$} -- (1.83,1);
    		
    		\draw[black,double,thick] (7,0) -- (7,0.645497);
    		\draw[black] (7,1) -- (7,0.8333333);
    		\draw[black,double,thick] (6,0) -- (6,0.632456);
    		\draw[black] (6,1) -- (6,0.8);
    		\draw[black,double,thick] (5,0) -- (5,0.612372);
    		\draw[black] (5,1) -- (5,0.75);
    		\draw[black,double,thick] (4,0) -- (4,0.57735);
    		\draw[black] (4,1) -- (4,0.6666666);
    		\draw[black,double,thick] (3,0) -- (3,0.5);
    		\draw[black] (3,1) -- (3,0.5);
    		\draw[black] (2,1) -- (2,0);

    		\filldraw[black] (2,0) circle (1pt);
    		\filldraw[black] (3,0.5) circle (1pt);
    		\filldraw[draw=black,fill=white] (4,0.6666666) circle (1pt);
    		\filldraw[draw=black,fill=white] (5,0.75) circle (1pt);
    		\filldraw[draw=black,fill=white] (6,0.8) circle (1pt);
    		\filldraw[draw=black,fill=white] (7,0.8333333) circle (1pt);
    		\filldraw[black] (4,0.57735) circle (1pt);
    		\filldraw[black] (5,0.612372) circle (1pt);
    		\filldraw[black] (6,0.632456) circle (1pt);
    		\filldraw[black] (7,0.645497) circle (1pt);
    	\end{tikzpicture}
	\end{center}
	\caption{How inner products between $n$ pure states determine their antidistinguishability. If all inner products are on or below the filled-in circles, then the states are antidistinguishable (Corollary~\ref{cor:antidist_by_IP}), and if all inner products are strictly above the hollow circles, then the states are not antidistinguishable (Corollary~\ref{cor:large_ip_not_anti}). In between those circles, the states might be antidistinguishable (Example~\ref{exam:equi_angle}), and when $n = 4$ at least they might also be not antidistinguishable (Example~\ref{ex:d4_lb_tight}).}\label{fig:antidist_ip_range}
\end{figure}

We still do not know whether or not Corollary~\ref{cor:antidist_by_IP} is tight when $n \geq 5$. The difficulty here is that the circulant matrices that we used to show that Corollary~\ref{cor:antidist_by_frob} is tight only have the property that their off-diagonal entries all have absolute value equal to each other when $n \leq 3$. When $n \geq 4$ we must explore non-circulant matrices like the one from Example~\ref{ex:d4_lb_tight}, and this seems much more difficult. 

\bigskip 
\noindent \textbf{Software.}  
Companion software that implements the SDPs from Equation~\eqref{eq:antidist-sdp} in addition to Examples~\ref{exam:trine}, \ref{exam:equi_angle}, and~\ref{ex:d4_lb_tight} can be found at the GitHub repository~\cite{russo2023circulant}. This repository contains Python code that makes use of the toqito quantum information package~\cite{russo2021toqito} as well as the PICOS package~\cite{sagnol2022picos} which invokes the CVXOPT solver~\cite{andersen2020cvxopt} for solving the SDPs.

\bigskip 
\noindent \textbf{Acknowledgements.} 
The authors thank Mark Hamilton, who provided a key insight that significantly simplified the proof of Corollary~\ref{cor:antidist_by_frob}, as well as Debbie Leung, Robert Spekkens, Iman Marvian, and Vojt\v{e}ch Havl\'{i}\v{c}ek for helpful conversations. N.J.\ was supported by NSERC Discovery Grant RGPIN-2022-04098. 
J.S. is partially supported by
Commonwealth Cyber Initiative SWVA grant 467489.  

\bibliographystyle{quantum} 
\bibliography{main} 

\begin{thebibliography}{10}

\bibitem{caves2002conditions}
Carlton~M. Caves, Christopher~A. Fuchs, and Rüdiger Schack.
\newblock ``Conditions for compatibility of quantum-state assignments''.
\newblock \href{https://dx.doi.org/10.1103/PhysRevA.66.062111}{Physical Review
  A {\bf 66}, 062111}~(2002).

\bibitem{leifer2014quantum}
Matthew~Saul Leifer.
\newblock ``Is the quantum state real? {A}n extended review of $\psi$-ontology
  theorems''~(2014).

\bibitem{heinosaari2018antidistinguishability}
Teiko Heinosaari and Oskari Kerppo.
\newblock ``Antidistinguishability of pure quantum states''.
\newblock \href{https://dx.doi.org/10.1088/1751-8121/aad1fc}{Journal of Physics
  A: Mathematical and Theoretical {\bf 51}, 365303}~(2018).

\bibitem{pusey2012reality}
Matthew~F. Pusey, Jonathan Barrett, and Terry Rudolph.
\newblock ``On the reality of the quantum state''.
\newblock \href{https://dx.doi.org/10.1038/nphys2309}{Nature Physics {\bf 8},
  475--478}~(2012).

\bibitem{bandyopadhyay2014conclusive}
Somshubhro Bandyopadhyay, Rahul Jain, Jonathan Oppenheim, and Christopher
  Perry.
\newblock ``Conclusive exclusion of quantum states''.
\newblock \href{https://dx.doi.org/10.1103/PhysRevA.89.022336}{Physical Review
  A {\bf 89}, 022336}~(2014).

\bibitem{perry2015communication}
Christopher Perry, Rahul Jain, and Jonathan Oppenheim.
\newblock ``Communication tasks with infinite quantum-classical separation''.
\newblock \href{https://dx.doi.org/10.1103/PhysRevLett.115.030504}{Physical
  Review Letters {\bf 115}, 030504}~(2015).

\bibitem{heinosaari2019communication}
Teiko Heinosaari and Oskari Kerppo.
\newblock ``Communication of partial ignorance with qubits''.
\newblock \href{https://dx.doi.org/10.1088/1751-8121/ab3ae4}{Journal of Physics
  A: Mathematical and Theoretical {\bf 52}, 395301}~(2019).

\bibitem{havlivcek2020simple}
Vojt{v{e}}ch Havl{\'i}{v{c}}ek and Jonathan Barrett.
\newblock ``Simple communication complexity separation from quantum state
  antidistinguishability''.
\newblock \href{https://dx.doi.org/10.1103/PhysRevResearch.2.013326}{Physical
  Review Research {\bf 2}, 013326}~(2020).

\bibitem{dunjko2014quantum}
Vedran Dunjko, Petros Wallden, and Erika Andersson.
\newblock ``Quantum digital signatures without quantum memory''.
\newblock \href{https://dx.doi.org/10.1103/PhysRevLett.112.040502}{Physical
  Review Letters {\bf 112}, 040502}~(2014).

\bibitem{amiri2021imperfect}
Ryan Amiri, Robert Stárek, David Reichmuth, Ittoop~V. Puthoor, Michal Mičuda,
  Jr. Mišta, Ladislav, Miloslav Dušek, Petros Wallden, and Erika Andersson.
\newblock ``Imperfect 1-{Out}-of-2 {Quantum} {Oblivious} {Transfer}: {Bounds},
  a {Protocol}, and its {Experimental} {Implementation}''.
\newblock \href{https://dx.doi.org/10.1103/PRXQuantum.2.010335}{PRX Quantum
  {\bf 2}, 010335}~(2021).

\bibitem{bennett1999quantum}
Charles~H. Bennett, David~P. DiVincenzo, Christopher~A. Fuchs, Tal Mor, Eric
  Rains, Peter~W. Shor, John~A. Smolin, and William~K. Wootters.
\newblock ``Quantum nonlocality without entanglement''.
\newblock \href{https://dx.doi.org/10.1103/PhysRevA.59.1070}{Physical Review A
  {\bf 59}, 1070}~(1999).

\bibitem{chefles2000quantum}
Anthony Chefles.
\newblock ``Quantum state discrimination''.
\newblock \href{https://dx.doi.org/10.1080/00107510010002599}{Contemporary
  Physics {\bf 41}, 401--424}~(2000).

\bibitem{walgate2000local}
Jonathan Walgate, Anthony~J. Short, Lucien Hardy, and Vlatko Vedral.
\newblock ``Local distinguishability of multipartite orthogonal quantum
  states''.
\newblock \href{https://dx.doi.org/10.1103/PhysRevLett.85.4972}{Physical Review
  Letters {\bf 85}, 4972}~(2000).

\bibitem{ghosh2001distinguishability}
Sibasish Ghosh, Guruprasad Kar, Anirban Roy, Aditi Sen, and Ujjwal Sen.
\newblock ``Distinguishability of {B}ell states''.
\newblock \href{https://dx.doi.org/10.1103/PhysRevLett.87.277902}{Physical
  Review Letters {\bf 87}, 277902}~(2001).

\bibitem{virmani2001optimal}
Shashank Virmani, Massimiliano~F. Sacchi, Martin~B. Plenio, and Damian Markham.
\newblock ``Optimal local discrimination of two multipartite pure states''.
\newblock \href{https://dx.doi.org/10.1016/S0375-9601(01)00484-4}{Physics
  Letters A {\bf 288}, 62--68}~(2001).

\bibitem{walgate2002nonlocality}
Jonathan Walgate and Lucien Hardy.
\newblock ``Nonlocality, asymmetry, and distinguishing bipartite states''.
\newblock \href{https://dx.doi.org/10.1103/PhysRevLett.89.147901}{Physical
  Review Letters {\bf 89}, 147901}~(2002).

\bibitem{horodecki2003local}
Micha{\l} Horodecki, Aditi Sen, Ujjwal Sen, and Karol Horodecki.
\newblock ``Local indistinguishability: More nonlocality with less
  entanglement''.
\newblock \href{https://dx.doi.org/10.1103/PhysRevLett.90.047902}{Physical
  Review Letters {\bf 90}, 047902}~(2003).

\bibitem{watrous2005bipartite}
John Watrous.
\newblock ``Bipartite subspaces having no bases distinguishable by local
  operations and classical communication''.
\newblock \href{https://dx.doi.org/10.1103/PhysRevLett.95.080505}{Physical
  Review Letters {\bf 95}, 080505}~(2005).

\bibitem{barnett2009quantum}
Stephen~M. Barnett and Sarah Croke.
\newblock ``Quantum state discrimination''.
\newblock \href{https://dx.doi.org/10.1364/AOP.1.000238}{Advances in Optics and
  Photonics {\bf 1}, 238--278}~(2009).

\bibitem{bergou2010discrimination}
J{\'a}nos~A. Bergou.
\newblock ``Discrimination of quantum states''.
\newblock \href{https://dx.doi.org/10.1080/09500340903477756}{Journal of Modern
  Optics {\bf 57}, 160--180}~(2010).

\bibitem{qiu2010minimum}
Daowen Qiu and Lvjun Li.
\newblock ``Minimum-error discrimination of quantum states: Bounds and
  comparisons''.
\newblock \href{https://dx.doi.org/10.1103/PhysRevA.81.042329}{Physical Review
  A {\bf 81}, 042329}~(2010).

\bibitem{bae2015quantum}
Joonwoo Bae and Leong-Chuan Kwek.
\newblock ``Quantum state discrimination and its applications''.
\newblock \href{https://dx.doi.org/10.1088/1751-8113/48/8/083001}{Journal of
  Physics A: Mathematical and Theoretical {\bf 48}, 083001}~(2015).

\bibitem{mishra2023optimal}
Hemant~K. Mishra, Michael Nussbaum, and Mark~M. Wilde.
\newblock ``On the optimal error exponents for classical and quantum
  antidistinguishability''~(2023).

\bibitem{PhysRevA.101.062113}
Matthew Leifer and Cristhiano Duarte.
\newblock ``Noncontextuality inequalities from antidistinguishability''.
\newblock \href{https://dx.doi.org/10.1103/PhysRevA.101.062113}{Phys. Rev. A
  {\bf 101}, 062113}~(2020).

\bibitem{PRXQuantum.3.020307}
Pierre-Emmanuel Emeriau, Mark Howard, and Shane Mansfield.
\newblock ``Quantum advantage in information retrieval''.
\newblock \href{https://dx.doi.org/10.1103/PRXQuantum.3.020307}{PRX Quantum
  {\bf 3}, 020307}~(2022).

\bibitem{montina2014lower}
Alberto Montina and Stefan Wolf.
\newblock ``Lower bounds on the communication complexity of two-party (quantum)
  processes''.
\newblock In 2014 IEEE International Symposium on Information Theory.
\newblock \href{https://dx.doi.org/10.1109/ISIT.2014.6875080}{Pages
  1484--1488}.
\newblock IEEE~(2014).

\bibitem{montina2014necessary}
Alberto Montina and Stefan Wolf.
\newblock ``Necessary and sufficient optimality conditions for classical
  simulations of quantum communication processes''.
\newblock \href{https://dx.doi.org/10.1103/PhysRevA.90.012309}{Physical Review
  A {\bf 90}, 012309}~(2014).

\bibitem{russo2023inner}
Vincent Russo and Jamie Sikora.
\newblock ``Inner products of pure states and their antidistinguishability''.
\newblock \href{https://dx.doi.org/10.1103/PhysRevA.107.L030202}{Physical
  Review A {\bf 107}, L030202}~(2023).

\bibitem{sentis2016quantum}
Gael Sent{\'\i}s, Emilio Bagan, John Calsamiglia, Giulio Chiribella, and Ramon
  Munoz-Tapia.
\newblock ``Quantum change point''.
\newblock \href{https://dx.doi.org/10.1103/PhysRevLett.117.150502}{Physical
  Review Letters {\bf 117}, 150502}~(2016).

\bibitem{sentis2017exact}
Gael Sent{\'\i}s, John Calsamiglia, and Ramon Munoz-Tapia.
\newblock ``Exact identification of a quantum change point''.
\newblock \href{https://dx.doi.org/10.1103/PhysRevLett.119.140506}{Physical
  Review Letters {\bf 119}, 140506}~(2017).

\bibitem{RBC18}
Martin Ringbauer, Thomas~R. Bromley, Marco Cianciaruso, Ludovico Lami,
  W.~Y.~Sarah Lau, Gerardo Adesso, Andrew~G. White, Alessando Fedrizzi, and
  Marco Piani.
\newblock ``Certification and quantification of multilevel quantum coherence''.
\newblock \href{https://dx.doi.org/10.1103/PhysRevX.8.041007}{Physical Review X
  {\bf 8}, 041007}~(2018).

\bibitem{LBT19}
Zi-Wen Liu, Kaifeng Bu, and Ryuji Takagi.
\newblock ``One-shot operational quantum resource theory''.
\newblock \href{https://dx.doi.org/10.1103/PhysRevLett.123.020401}{Physical
  Review Letters {\bf 123}, 020401}~(2019).

\bibitem{LSLL21}
Jun-Wei Liu, Shu-Qian Shen, Ming Li, and Lei Li.
\newblock ``Lower bounds for the robustness of multilevel coherence''.
\newblock \href{https://dx.doi.org/10.1007/s10773-021-04793-1}{International
  Journal of Theoretical Physics {\bf 60}, 1712--1719}~(2021).

\bibitem{LM14}
Federico Levi and Florian Mintert.
\newblock ``A quantitative theory of coherent delocalization''.
\newblock \href{https://dx.doi.org/10.1088/1367-2630/16/3/033007}{New Journal
  of Physics {\bf 16}, 033007}~(2014).

\bibitem{ZGY21}
LiMei Zhang, Ting Gao, and FengLi Yan.
\newblock ``Transformations of multilevel coherent states under
  coherence-preserving operations''.
\newblock \href{https://dx.doi.org/10.1007/s11433-021-1696-y}{Science China
  Physics, Mechanics \& Astronomy {\bf 64}, 1--6}~(2021).

\bibitem{JMPP22}
Nathaniel Johnston, Shirin Moein, Rajesh Pereira, and Sarah Plosker.
\newblock ``Absolutely k-incoherent quantum states and spectral inequalities
  for factor width of a matrix''.
\newblock \href{https://dx.doi.org/10.1103/PhysRevA.106.052417}{Physical Review
  A {\bf 106}, 052417}~(2022).

\bibitem{NC00}
Michael~A. Nielsen and Isaac~L. Chuang.
\newblock ``{Quantum Computation and Quantum Information}''.
\newblock \href{https://dx.doi.org/10.1017/CBO9780511976667}{{Cambridge
  University Press}}. ~(2000).

\bibitem{watrous2018theory}
John Watrous.
\newblock ``{The Theory of Quantum Information}''.
\newblock \href{https://dx.doi.org/10.1017/9781316848142}{{Cambridge University
  Press}}. ~(2018).

\bibitem{graeme2018thesis}
Graeme Weir.
\newblock ``Optimal discrimination of quantum states''.
\newblock PhD thesis.
\newblock University of Glasgow.
\newblock ~(2018).

\bibitem{JohALA}
Nathaniel Johnston.
\newblock ``Advanced linear and matrix algebra''.
\newblock \href{https://dx.doi.org/10.1007/978-3-030-52815-7}{Springer}.
  ~(2021).

\bibitem{CircBook}
Philip~J. Davis.
\newblock ``Circulant matrices''.
\newblock Monographs and textbooks in pure and applied mathematics. Wiley.
  ~(1979).

\bibitem{DA12}
Vedran Dunjko and Erika Andersson.
\newblock ``Transformations between symmetric sets of quantum states''.
\newblock \href{https://dx.doi.org/10.1088/1751-8113/45/36/365304}{Journal of
  Physics A: Mathematical and Theoretical {\bf 45}, 365304}~(2012).

\bibitem{howard2012qudit}
Mark Howard and Jiri Vala.
\newblock ``Qudit versions of the qubit $\pi$/8 gate''.
\newblock \href{https://dx.doi.org/10.1103/PhysRevA.86.022316}{Physical Review
  A—Atomic, Molecular, and Optical Physics {\bf 86}, 022316}~(2012).

\bibitem{dalla2015optimality}
Nicola Dalla~Pozza and Gianfranco Pierobon.
\newblock ``Optimality of square-root measurements in quantum state
  discrimination''.
\newblock \href{https://dx.doi.org/10.1103/PhysRevA.91.042334}{Physical Review
  A {\bf 91}, 042334}~(2015).

\bibitem{DM16}
Louis Deaett and Seth~A. Meyer.
\newblock ``The minimum rank problem for circulants''.
\newblock \href{https://dx.doi.org/10.1016/j.laa.2015.10.033}{Linear Algebra
  and its Applications {\bf 491}, 386--418}~(2016).

\bibitem{Spe15}
Jan Sperling and Werner Vogel.
\newblock ``Convex ordering and quantification of quantumness''.
\newblock \href{https://dx.doi.org/10.1088/0031-8949/90/7/074024}{Physica
  Scripta {\bf 90}, 074024}~(2015).

\bibitem{boman2005factor}
Erik~G. Boman, Doron Chen, Ojas Parekh, and Sivan Toledo.
\newblock ``On factor width and symmetric {H}-matrices''.
\newblock \href{https://dx.doi.org/10.1016/j.laa.2005.03.029}{Linear Algebra
  and its Applications {\bf 405}, 239--248}~(2005).

\bibitem{blekherman2022hyperbolic}
Grigoriy Blekherman, Santanu~S Dey, Kevin Shu, and Shengding Sun.
\newblock ``Hyperbolic relaxation of k-locally positive semidefinite
  matrices''.
\newblock \href{https://dx.doi.org/10.1137/20M1387407}{SIAM Journal on
  Optimization {\bf 32}, 470--490}~(2022).

\bibitem{BV04}
Stephen Boyd and Lieven Vandenberghe.
\newblock ``Convex optimization''.
\newblock \href{https://dx.doi.org/10.1017/CBO9780511804441}{Cambridge
  University Press}. ~(2004).

\bibitem{Zin08}
Yuriy Zinchenko.
\newblock ``On hyperbolicity cones associated with elementary symmetric
  polynomials''.
\newblock \href{https://dx.doi.org/10.1007/s11590-007-0067-0}{Optimization
  Letters {\bf 2}, 389--402}~(2008).

\bibitem{russo2023circulant}
Nathaniel Johnston, Vincent Russo, and Jamie Sikora.
\newblock ``circulant\_antidist: A {P}ython toolkit for studying the
  antidistinguishability of circulant pure states''.
\newblock \url{https://github.com/vprusso/circulant_antidist}~(2023).

\bibitem{russo2021toqito}
Vincent Russo.
\newblock ``toqito -- {Theory} of quantum information toolkit: {A} {Python}
  package for studying quantum information''.
\newblock \href{https://dx.doi.org/10.21105/joss.03082}{Journal of Open Source
  Software {\bf 6}, 3082}~(2021).

\bibitem{sagnol2022picos}
Guillaume Sagnol and Maximilian Stahlberg.
\newblock ``{PICOS}: {A} {Python} interface to conic optimization solvers''.
\newblock \href{https://dx.doi.org/10.21105/joss.03915}{Journal of Open Source
  Software {\bf 7}, 3915}~(2022).

\bibitem{andersen2020cvxopt}
Martin Andersen, Joachim Dahl, and Lieven Vandenberghe.
\newblock ``{CVXOPT}: {Convex} {Optimization}''.
\newblock Astrophysics Source Code LibraryPage ascl:2008.017~(2020).
\newblock  url:~\url{https://ui.adsabs.harvard.edu/abs/2020ascl.soft08017A}.

\end{thebibliography}

\end{document}